\documentclass{amsart}

\usepackage[english]{babel}
\usepackage[utf8]{inputenc}
\usepackage[T1]{fontenc}

\usepackage{amsmath}
\usepackage{amssymb}
\usepackage{mathtools}
\usepackage{amsthm}
\usepackage{physics}
\usepackage{cite}
\usepackage{enumitem}

\usepackage{booktabs}
\usepackage{rotating}
\usepackage{tabularx}
\usepackage{multirow}

\usepackage{siunitx}

\usepackage{graphicx}
\usepackage{subfig}
\usepackage{wrapfig}
\usepackage{xcolor}

\usepackage{comment}

\usepackage{etoolbox} % for '\AtBeginEnvironment' macro
\AtBeginEnvironment{pmatrix}{\everymath{\displaystyle}}
\AtBeginEnvironment{bmatrix}{\everymath{\displaystyle}}

\DeclareMathOperator{\dpat}{dp}
\DeclareMathOperator{\arctanh}{arctanh}

\newcommand{\Ilow}{I_\text{low}}
\newcommand{\Ihigh}{I_\text{high}}
\newcommand{\N}{\mathbb{N}}
\newcommand{\Pj}{\mathbb{CP}}
\newcommand{\Cp}{\mathbb{C}}
\newcommand{\dP}{\mathrm{d}P}
\newcommand{\pvec}[1]{\mathbf{[#1]}}
\renewcommand{\vec}[1]{\mathbf{#1}}
\DeclareMathOperator{\ud}{d}

\newcounter{eqlist}
\renewcommand*\theeqlist{\Roman{eqlist}}
\newcounter{eqlistred}

\allowdisplaybreaks
\makeatletter 
\renewcommand{\pdv}[2]{\begingroup 
\@tempswafalse\toks@={}\count@=\z@ 
\@for\next:=#2\do 
{\expandafter\check@var\next\@nil
 \advance\count@\der@exp 
 \if@tempswa 
   \toks@=\expandafter{\the\toks@\,}% 
 \else 
   \@tempswatrue 
 \fi 
 \toks@=\expandafter{\the\expandafter\toks@\expandafter\partial\der@var}}% 
\frac{\partial\ifnum\count@=\@ne\else^{\number\count@}\fi#1}{\the\toks@}% 
\endgroup} 
\def\check@var{\@ifstar{\mult@var}{\one@var}} 
\def\mult@var#1#2\@nil{\def\der@var{#2^{#1}}\def\der@exp{#1}} 
\def\one@var#1\@nil{\def\der@var{#1}\chardef\der@exp\@ne} 
\makeatother

\theoremstyle{plain}
\newtheorem{theorem}{Theorem}

\newtheorem{corollary}[theorem]{Corollary}
\newtheorem{lemma}[theorem]{Lemma}
\theoremstyle{definition}
\newtheorem{definition}{Definition}
\newtheorem{problem}{Problem}
\theoremstyle{remark}
\newtheorem{remark}{Remark}
\theoremstyle{remark}
\newtheorem{example}{Example}

\begin{document}
\title[Bi-rational mappings in 4D]{Bi-rational maps in four dimensions with two invariants}
\date{\today}
\author[G. Gubbiotti]{Giorgio Gubbiotti}
\author[N. Joshi]{Nalini Joshi}
\author[D. T. Tran]{Dinh Thi Tran}
\address{School  of  Mathematics  and  Statistics  F07,  The  University  of  Sydney,  NSW  2006, Australia}
\email{giorgio.gubbiotti@syndey.edu.au}
\email{nalini.joshi@syndey.edu.au}
\email{dinhthi.tran@syndey.edu.au}
\author[C-M. Viallet]{Claude-Michel Viallet}
\address{LPTHE, UMR 7589 Centre National de la Recherche Scientifique \& UPMC Sorbonne Universit\'es, 4 place Jussieu, 75252 Paris Cedex 05, France}
\email{claude.viallet@upmc.fr}
\begin{abstract}
    In this paper we present a class of four-dimensional bi-rational 
    maps with two invariants satisfying certain constraints on degrees.
    We discuss the integrability properties of these maps from
    the point of view of degree growth and Liouville integrability.
\end{abstract}

\maketitle

\section{Introduction}

In this paper we classify maps of $\Pj^{4}$ to itself, which possess two
polynomial invariants, under certain conditions.
The outcomes of our classification include eight new classes of maps,
which have surprising properties.
Despite the existence of two invariants, there turn out to be
non-integrable cases, with exponential growth.
Other cases are integrable, with cubic and quadratic growth.
The cases of cubic growth are only possible in dimension
greater than two \cite{Diller1996,DillerFavre2001}. 
We discuss the geometric properties of these systems
\cite{Bellon1999}.

In dimension two it is  well-known that integrable bi-rational 
maps can be characterized by the existence of a rational invariant.
For instance most of the integrable maps on the plane fall in the class
of QRT maps \cite{QRT1988,QRT1989}, even though there are
some notable exceptions 
\cite{VialletRamaniGrammaticos2004,Duistermaat2011book, RJ15}.
The integrability of these maps can be explained geometrically
and has led to many interesting developments
\cite{Sakai2001,Duistermaat2011book,Tsuda2004}. 

In higher dimension an analogous general framework does not exist.
In particular, for mappings in four dimension, a generalisation of the 
QRT class \cite{QRT1988,QRT1989} was given in \cite{CapelSahadevan2001}. 
However, this generalisation does not cover all possible integrable 
maps in four dimensions.
Indeed, some of the new maps obtained in \cite{CapelSahadevan2001} 
turn out to be autonomous versions of Painlev\'e hierarchies \cite{Hay2007} 
which are \emph{multiplicative} equations in Sakai's scheme \cite{Sakai2001}.
On the other hand,  there  exists hierarchies of \emph{additive}
discrete Painlev\'e equations too \cite{CresswellJoshi1999}.
Equations coming from the hierarchies of additive Painlev\'e equations
are naturally outside the framework of \cite{CapelSahadevan2001}. 
Other examples of four-dimensional maps falling outside the class presented in 
\cite{CapelSahadevan2001} are given in
\cite{JoshiViallet2017,GJTV_sanya,PetreraSuris2010,CelledoniMcLachlanOwrenQuispel2013,
CelledoniMcLachlanOwrenQuispel2014,PetreraPfadlerSuris2009}.

Our starting point is \cite{JoshiViallet2017}, where the authors considered the
\emph{autonomous limit} of the second member of the 
$\dP_\text{I}$ and $\dP_\text{II}$ hierarchies \cite{CresswellJoshi1999}.
We will denote these equations as $\dP_\text{I}^{(2)}$ and $\dP_\text{II}^{(2)}$ equations.
These $\dP_\text{I}^{(2)}$ and $\dP_\text{II}^{(2)}$ equations
are given by recurrence relations of order four,
and shown to be integrable
according to the algebraic entropy approach.
Therein the authors showed that both maps possess two polynomial invariants.
Using these invariants, they produced the dual maps of the 
$\dP_\text{I}^{(2)}$ and $\dP_\text{II}^{(2)}$ equations in the sense of
\cite{QuispelCapelRoberts2005}.
Moreover, they showed that these dual maps are integrable according to the 
algebraic entropy test and also possess invariants.
In fact, the number of invariants showed that the dual maps are 
actually \emph{superintegrable}.
Finally they gave a scheme to construct recurrence relations of an assigned form.
Using this scheme in \cite{JoshiViallet2017} some new examples, with no classification
purposed were presented.
Starting from these considerations in this paper we consider and solve 
the problem of finding  all fourth-order bi-rational maps  possessing two polynomial invariants 
of general enough form to contain those of the $\dP_\text{I}^{(2)}$ and  
$\dP_\text{II}^{(2)}$ equations.

The structure of the paper is the following:
in section \ref{sec:setting} we give a concise explanation of the background
material we need. In particular we discuss the various definitions
of integrability for mapping we are going to use throughout the paper.
In section \ref{sec:theproblem} we present the motivations for our search
and we present our the search method and we state the general result. 
In section \ref{sec:theclass} we give the explicit form of the maps
we derived with the method of section \ref{sec:theproblem} and we discuss
their integrability properties following the discussion of section 
\ref{sec:setting}.
Finally, in section \ref{sec:conclusions} we make some general comments on
the maps we obtained, and we underline the possible future development.

\section{Setting}
\label{sec:setting}

In this Section we give the fundamental definitions we need to explain
how our list of equations is found and what kind of integrability
we are going to consider within this paper. One can also find this setting in our shot communication \cite{GJTV_sanya}.

\subsection{Bi-rational maps and invariants}

%In this paper, we focus on the study of integrability
%properties of recurrence relations.
%In this paper we define an autonomous recurrence relation to be
%a bi-rational map of the complex projective space into itself:
The main subject of this paper are \emph{bi-rational maps} of
the complex projective space into itself:
\begin{equation}
    \varphi \colon \pvec{x}\in\Pj^{n} \to \pvec{x'} \in \Pj^{n},
    \label{eq:mapp4}
\end{equation}
where $n>1$\footnote{Bi-rational maps in $\Pj^{1}$ are just M\"obius
transformations so everything is trivial.} and  $\pvec{x}=\left[ x_{1}:x_{2}:\dots:x_{n+1} \right]$ 
and $\pvec{x'}=\left[ x_{1}':x_{2}':\dots:x_{n+1}' \right]$
to be homogeneous coordinates on $\Pj^{n}$.
Moreover, we recall that a bi-rational map is a rational map
$\varphi\colon V\to W$ of algebraic varieties $V$ and $W$ such that there exists a
rational map $\psi\colon W\to V$, which is the \emph{inverse} of $\varphi$ in the
dense subset where both maps are defined \cite{Shafarevich1994}.

Bi-rational maps of the form \eqref{eq:mapp4} are the natural mathematical
object needed to study \emph{autonomous single-valued invertible rational 
recurrence relations}.
Indeed, an autonomous recurrence relation of order $n$ is a relation
where the $(n+1)$-th element of a sequence is defined in terms of the preceding
$n$, i.e. an expression of the form:
\begin{equation}
    w_{k+n} = f\left( w_{k},\dots,w_{k+n-1} \right).
    \label{eq:recrelgen}
\end{equation}
A recurrence relation is \emph{autonomous} if the function $f$ in
\eqref{eq:recrelgen} does not depend explicitly on $k$.
Moreover,  we say that a recurrence relation is \emph{rational} if 
the function $f$ in \eqref{eq:recrelgen} is a rational function.
Finally, the recurrence is invertible and single-valued if equation 
\eqref{eq:recrelgen} is solvable uniquely with respect to $w_{0}$.
All the terms of the sequence $w_{k}$ for $k>n$ are then obtained by
iterated substitution of the previous one into equation \eqref{eq:recrelgen}.
For this reason it is possible to interpret the recurrence relation
\eqref{eq:recrelgen} as a map of the complex space of dimension $n$ into
itself as:
\begin{equation}
    \varpi \colon \vec{w}\in\Cp^{n} \to \vec{w'} \in \Cp^{n},
    \label{eq:recrelmapgen}
\end{equation}
where $\vec{w}=\left(w_{n},w_{n-1},\dots,w_{0}  \right)$ are the
initial conditions and the map acts as:
\begin{equation}
    \vec{w'} = \left(f\left(\vec{w} \right),w_{n},w_{n-1},\dots,w_{1}\right).
    \label{eq:recrelgenmapexpl}
\end{equation}
The recurrence relation \eqref{eq:recrelgen} is then given by 
the repeated application of the map $\varpi$, namely $w_{n+k}$ is the
first component of  $\varpi^{k}$.
Interpreting the coordinates $\vec{w}\in\Cp^{n}$ as an affine chart
in $\Pj^{n}$, i.e. assuming that 
$\left( w_{n-1},\dots,w_{0} \right)=\left[ w_{n-1}:\dots:w_{0}:1 \right]$
we have that the map \eqref{eq:recrelmapgen} can be brought to a bi-rational
map of $\Pj^{n}$ into itself of the form \eqref{eq:mapp4}.

Throughout the paper we will often make use of the correspondence between
bi-rational maps and recurrence relations.
This is due to the fact that some definitions are
easier to state and use in the projective setting, while others are easier to
state and use in the affine one.
In any case for us ``bi-rational map'' and ``recurrence relation'' will
be completely equivalent terms.

%Interpreting the coordinates $\vec{w}=\left(w_{n-1},\dots,w_{0}  \right)$
%as 

One characteristic of integrability is the existence of \emph{first integrals}.
In the continuous context,  for finite dimensional systems, integrability refers to 
 the existence of a ``sufficiently'' high number of 
first integrals, i.e. of \emph{non-trivial} functions constant
along the solution of the differential system. In particular, for 
a Hamiltonian system,  the number of first integrals is less  as its integrability 
was given by Liouville \cite{Liouville1855}.
In discrete setting, the analogue of 
first integrals  for maps are the \emph{invariants} which is defined as follows. 
\begin{definition}
    An \emph{invariant} of a bi-rational map $\varphi\colon\Pj^{n}\to\Pj^{n}$
    is a homogeneous function $I \colon \Pj^{n}\to \Cp$ such that
    it is preserved under the action of the map, i.e. 
    \begin{equation}
        \varphi^{*}\left( I \right) = I,
        \label{eq:fintdef}
    \end{equation}
    where $\varphi^{*}\left( I \right)$ means the pullback of $I$
    through the map $\varphi$, i.e. 
    $\varphi^{*}\left( I \right)=I\left( \varphi\left( \pvec{x} \right) \right)$.
    For $n>1$, an invariant is said to be \emph{non-degenerate} if:
    \begin{equation}
        \pdv{I}{x_{1}}\pdv{I}{x_{n}}\neq 0.
        \label{eq:nondegcond}
    \end{equation}
    Otherwise an invariant is said to be \emph{degenerate}.
    \label{def:inv}
\end{definition}

%\begin{remark}
%    Degenerate invariants are not well posed since they
%    cannot be evaluated either on the map $\varphi$ or on
%    its inverse $\varphi^{-1}$.
%    For this reason throughout this paper we will always assume
%    that the non-degeneracy condition \eqref{eq:nondegcond}
%    always holds.
%    \label{rem:deg}
%\end{remark}

In what follows we will concentrate on a particular class 
of invariants:

\begin{definition}
    An invariant $I$ is said to be \emph{polynomial}, if in the
    affine chart $\left[ x_{1}:\dots:x_{n}:1 \right]$ the function
    $I$ is a polynomial function.
    \label{def:polyinv}
\end{definition}

In definition \ref{def:polyinv} we use $x_{n+1}$ as homogenising variable 
to go from an affine (polynomial) form to a projective (rational) form of the
invariants.
A polynomial invariant in the sense of definition \ref{def:polyinv}
written in homogeneous variables is always a rational function
homogeneous of degree 0.
The form of the polynomial invariant in homogeneous coordinates
is then given by:
\begin{equation}
    I\left( \pvec{x} \right) =\frac{ I'\left( \pvec{x} \right)}{x_{n+1}^{d}},
    \quad d = \deg I'\left( \pvec{x}\right),
    \label{eq:polyinvrat}
\end{equation}
where $\deg$ is the total degree.

%In what follows we will consider \emph{polynomial invariants},
%i.e. invariants such that the function $I$ is a polynomial.
To better characterize the properties of these invariants we introduce
the following:
\begin{definition}
    %and hence we need to introduce the concept of \emph{degree pattern}
    Given a polynomial function $F\colon \Pj^{n}\to V$, where $V$
    can be either $\Pj^{n}$ or $\Cp$, we define the \emph{degree pattern} 
    of $F$ to be:
    \begin{equation}
        \dpat F = \left( \deg_{x_{1}} F,\deg_{x_{2}} F,\dots,\deg_{x_{n}} F\right).
        \label{eq:dpatt}
    \end{equation}
    \label{def:dpat}
\end{definition}

\begin{comment}
\begin{example}
    Consider the following map in $\Pj^{2}$:
    \begin{equation}
        \varphi\colon [ x:y:t] \mapsto [-y ( {x}^{2}-{t}^{2})+2 ax{t}^{2}: x ( {x}^{2}-{t}^{2})   : t ( {x}^{2}-{t}^{2})]
        \label{eq:mcm}
    \end{equation}
    This map is known as the \emph{McMillan map} \cite{McMillan1971} and
    possesses the following invariant:
    \begin{equation}
        t^{4} I_{\text{McM}} = x^2y^2+( x^2+y^2-2 a x y ) t^2
        \label{eq:intmcm}
    \end{equation}
    We have $\dpat I_{\text{McM}}=\left( 2,2 \right)$, i.e.
    it is a \emph{bi-quadratic} polynomial.
    We also note that the invariant of a QRT map \cite{QRT1988,QRT1989}, $I_\text{QRT}$, 
    which is a generalization of the McMillan map \eqref{eq:mcm},
    is the ratio of two \emph{bi-quadratic} in the dynamical variables of $\Pj^{2}$.
    Hence QRT mappings leave invariant a pencil of curves 
    of degree pattern $\left( 2,2 \right)$.
\end{example}

\begin{example}
    The invariants of the maps presented in
    \cite{CapelSahadevan2001}, $I_\text{CS}$, are 
    are ratios of \emph{bi-quadratic} in all the four dynamical variables
    of $\Pj^{4}$, i.e. ratios of polynomial of degree pattern $\left( 2,2,2,2 \right)$. 
    In this sense the classification of \cite{CapelSahadevan2001}
    is an extension of the one in \cite{QRT1988,QRT1989}.
\end{example}
\end{comment}

Finally we will consider invariants which are not
of generic shape, but satisfy the  following condition: 
\begin{definition}
    We say that a  function  $I\colon\Pj^{n}\to\Cp$
    is \emph{symmetric} if it is invariant under 
    the following involution:
    \begin{equation}
        \iota \colon \left[ x_{1}:x_{2}:\dots:x_{n}:x_{n+1} \right]
        \to
        \left[ x_{n}:x_{n-1}:\dots:x_{1}:x_{n+1} \right],
        \label{eq:symmcond}
    \end{equation}
    \label{def:symm}
    i.e. $\iota^{*}\left( I \right)=I$.
\end{definition}

\subsection{Integrability of bi-rational maps}

Integrability both for continuous and discrete systems can be defined
in different ways, see \cite{whatisintegrability1991,HietarintaBook}
for a complete discussion of the continuous and the discrete case.
Different ways of defining integrability do not always necessarily 
agree, even though most of the time they do.
We underline that the list we are going to make is not meant to be 
completely exhaustive of all the possible definitions of integrability.
We will discuss only the definitions for autonomous recurrence 
relations we will need throughout the rest of the paper.
We mention that additional definitions of integrability 
have been proposed for non-autonomous systems.

In general the solution of a recurrence relation of order
$n$ will depend on $n$ arbitrary constants.
This means that if a recurrence relation defined by the map 
$\varphi\colon\Pj^{n}\to\Pj^{n}$ possesses $n-1$ invariants
$I_{j}$, $j=1,\dots,n-1$, then, in principle, it is possible
to reduce it to a map $\hat{\varphi}\colon\Pj^{1}\to\Pj^{1}$
by solving the relations:
\begin{equation}
    I_{j} = \kappa_{j},
    \label{eq:intred}
\end{equation}
where $\kappa_{j}$ are the value of the invariants on a set of initial
data.
This stimulates the simplest and most natural definition of 
integrability for maps:

\begin{definition}[Existence of invariants]
    An $n$-dimensional map is (super)integrable if 
    \emph{it admits $n-1$ functionally independent invariants}.
    \label{def:invint}
\end{definition}

\begin{remark}
    We underline that, in general, the reduction to a lower-dimensional map 
    solving the system of equations \eqref{eq:intred} can break the bi-rationality.
    \label{rem:birat}
\end{remark}

Definition \ref{def:invint} is very general, and works for arbitrary maps.
If some additional structure are present, then the number of invariants
needed for integrability can be significantly reduced.
A special, but relevant case is the one of Poisson maps.

\begin{definition}[Poisson structures and Poisson maps \cite{CapelSahadevan2001, Olver1986}]
    In affine coordinates 
    $\vec{w}$
    a \emph{Poisson structure of rank $2r$} is a skew-symmetric matrix $J=J\left( \vec{w} \right)$ 
    of constant rank $2r$ such that the \emph{Jacobi identity holds}:
    \begin{equation}
    \label{eq:JacoIden}
    \sum_{l=1}^{n}\left(J_{li}\frac{\partial J_{jk}}{\partial w_{l-1}}
        +J_{lj}\frac{\partial J_{ki}}{\partial w_{l-1}}
        +J_{lk}\frac{\partial J_{ij}}{\partial w_{l-1}}\right)=0,
        \quad
        \forall i,j,k.
    \end{equation}
    A Poisson structure defines a \emph{Poisson bracket} through
    the identity:
    \begin{equation}
    \label{eq:Poisson_def}
    \{f, g\}=\nabla f   J \left( \vec{w} \right)  \nabla g^{T},
    \end{equation}
    where $\nabla f$ is the gradient of $f$. 
    Two functions $f$ and $g$ are said to be \emph{in involution} with respect
    to the Poisson structure $J\left( \vec{w} \right)$ if $\left\{ f,g \right\}=0$.
    We can easily see that $\{w_{i-1},w_{j-1}\}=J_{ij}$.
    A map of the affine coordinates $\varphi\colon \vec{w}\mapsto\vec{w'}$ 
    is a \emph{Poisson map} if it preserves the Poisson structure 
    $J\left( \vec{w} \right)$, i.e. if:
    \begin{equation}
        \label{eq:Poisson1}
        \ud\varphi  J({\bf{w}})  \ud\varphi^T=J({\bf{w'}}),
    \end{equation}
    where $\ud\varphi$ is the  Jacobian matrix of the map $\varphi$. 
    \label{def:poisson}
\end{definition}

Then we have the following characterisation of integrability for
Poisson maps:

\begin{definition}[Liouville integrability \cite{Veselov1991,Bruschietal1991,Maeda1987}]
    An $n$-dimensional Poisson map is integrable if \emph{it possesses 
        $n-r$ functionally independent invariants in involution with respect to this
        Poisson structure}.
    \label{def:liouvilleintegrability}
\end{definition}

\begin{remark}
    A Poisson structure of full rank, i.e. $n=2r$ is invertible.
    The inverse matrix of the matrix $J\left( \vec{w} \right)$, i.e.
    $\Omega\left( \vec{w} \right)=J^{-1}\left( \vec{w} \right)$ is said to be 
    \emph{a symplectic structure}.
    We note that in the symplectic case  we only need $n/2$ 
    invariants in involution to claim integrability.
    \label{rem:symplectic}
\end{remark}

Symplectic structures are quite important in the theory of integrable maps. 
For instance, the classification made in \cite{CapelSahadevan2001}
was carried out assuming of the existence of \emph{linear}
Poisson structure and of two invariants.

A difficult problem is, given a map, to find if there exists a
symplectic structure for which this map is symplectic.
%To find whether a map is symplectic or not is a difficult problem.
In \cite{ByrnesHaggarQuispel1999} it was proved that there exists a 
\emph{pre-symplectic structure} (a degenerate sympectic structure)
for any $n$-dimensional volume-preserving map possessing $n-2$ invariants.
The rank of the obtained pre-symplectic structure is $n-2$ which implies
that to claim integrability in the sense of Liouville one must be able
to find another invariant.
On the other hand, when the map comes from a \emph{discrete variational principle}, 
i.e. it is \emph{variational}, to find a symplectic structure is easy.
We recall that an even-order recurrence relation \eqref{eq:recrelgen} 
is said to be variational if there exists a function, called \emph{Lagrangian}, 
$L=L\left( w_{k+N},\dots,w_{n} \right)$ such that the recurrence relation
\eqref{eq:recrelgen} is equivalent to the \emph{Euler-Lagrange} equations:
\begin{equation}
    \label{C4E:ELeq}
    \sum_{i=0}^{N} \frac{\partial L}{\partial w_{k}}\left( w_{k+N-i},\dots,w_{k-i} \right)
    =0.
    %\label{eq:elgen}
\end{equation}
Here $N=n/2$ in the recurrence \eqref{eq:recrelgen}.
A Lagrangian is called \emph{normal} if
\begin{equation}
    \pdv{L}{w_{k},w_{k+N}} \neq 0.
    \label{eq:normlagr}
\end{equation}
Let $T$ be a shift operator, i.e $T^j(w_{k+i})=w_{k+i+j}$.
Then due to the normality condition the 
\emph{discrete Ostrogradsky transformation} \cite{TranvanderKampQusipel2016}:
\begin{equation}
    \mathcal{O}\colon \vec{w} \to \left( \vec{q},\vec{p} \right)
    \label{eq:discrostr}
\end{equation}
where the new coordinates
$\left( \vec{q},\vec{p} \right)=\left( q_{1},\dots,q_{N},p_{1},\dots,p_{N} \right)$
are defined through the formula:
\begin{subequations}
    \begin{align}
        q_i&=w_{k+i-1}, \quad i=1,\dots,N,
        \label{C4E:OstraTranA} 
        \\
        p_i&=T^{-1}\sum_{j=0}^{N-i}T^{-j}\pdv{L}{w_{j+i}}, \quad
        i=1,\dots,N,
        \label{C4E:OstraTranB}
    \end{align}    
    \label{eq:C4EOstraTran}
\end{subequations}
is well defined and invertible.
Then the following result holds true \cite{Bruschietal1991}:
\begin{lemma}
    \label{L:mapcano}
    The map
    %\begin{equation}
    %    \Phi\colon\left( \vec{q},\vec{p} \right)\to\left( \vec{q'},\vec{p'} \right),
    %    \label{eq:canmap}
    %\end{equation}
    given by $\Phi=\mathcal{O}\circ \varpi\circ\mathcal{O}^{-1}\colon
    \left( \vec{q},\vec{p} \right)\to\left( \vec{q'},\vec{p'} \right) $, 
    where $\varpi$ is the map corresponding to the Euler-Lagrange equations
    \eqref{C4E:ELeq} has the following form
    \begin{subequations}
        \begin{align}
            q_i'&=q_{i+1}, \quad i=1,2,\ldots, N-1 , 
            \label{E:qi}
            \\
            q_N'&=\alpha(\vec{q},p_1),
            \label{E:qN}
            \\
            p_i' &= p_{i+1}+\pdv{\widetilde{L}}{q_{i+1}}(\vec{q},p_1),\ i=1,2,\ldots , N-1, 
            \label{E:pi}
            \\
            p_N'&=\left.\pdv{\widetilde{L}}{\tau}(q,p_1)\right|_{\tau = \alpha\left( \vec{q},p_{1} \right)},
            \label{E:pN}
        \end{align}
        \label{eq:canmap}
    \end{subequations}
    where $\alpha(\vec{q},p_{1})$ is the solution with respect to
    $q^{N'}$ of the equation:
    \begin{equation}
        p_1 =-\pdv{L}{q_{1}}(\vec{q},{q_N'}),
        \label{eq:qNpdef}
    \end{equation}
    and $\widetilde{L}(\vec{q},p_1)=L(\vec{q}, \alpha(\vec{q},p_1))$.
    Moreover, the map \eqref{eq:canmap} is symplectic with respect to
    the canonical symplectic structure:
    \begin{equation}
        \Omega = 
        \begin{pmatrix}
            \mathbb{O}_{N} & \mathbb{I}_{N}
            \\
            -\mathbb{I}_{N} & \mathbb{O}_{N}
        \end{pmatrix}
        \label{eq:Omegacan},
    \end{equation}
    where $\mathbb{O}_{N}$ is the zero $N\times N$ matrix
    and $\mathbb{I}_{N}$ the $N \times N$ identity matrix.
\end{lemma}

Lemma \ref{L:mapcano} has the following corollary:
\begin{corollary}
\label{C:Poisson}
    The Euler--Lagrange equations \eqref{C4E:ELeq} admit the
    following non-degenerate Poisson bracket:
    \begin{equation}
        J\left( \vec{w} \right) = \ud \mathcal{O}^{-1}\Omega^{-1}(\ud\mathcal{O}^{-1})^{T},
        \label{eq:canEL}
    \end{equation}
    where the differential of the Ostrogradsky transformation $\mathcal{O}$
    must be evaluated on the original coordinates.
    \label{cor:costrpb}
\end{corollary}

Therefore we have that corollary \ref{cor:costrpb} allows us to construct 
a non-degenerate Poisson structure, and hence a symplectic structure, 
\emph{for every variational map}.

Lagrangians for $2N$-order recurrence relations can be found following
\cite{HydonMansfield2004} or \cite{Gubbiotti_dcov} for $N>1$.
The method presented in \cite{Gubbiotti_dcov} allows also to disprove
the existence of a Lagragian for a given $2N$-order recurrence relation
for $N>1$.

Moreover, bi-rational maps possess another definition of integrability:
the \emph{low growth condition} \cite{Veselov1992,FalquiViallet1993,BellonViallet1999}.
To be more precise we say that an $n$-dimensional \emph{bi-rational}
map is integrable if \emph{the degree of growth of the iterated
map $\varphi^{k}$ is polynomial with respect to the initial conditions
$\pvec{x_{0}}$}.
Therefore we have the following characterisation of integrability:
\begin{definition}[Algebraic entropy \cite{BellonViallet1999}]
    An $n$-dimensional \emph{bi-rational}
    map is \emph{integrable in the sense of the algebraic entropy}
    if the following limit
    \begin{equation}
        \varepsilon = \lim_{k\to\infty}\frac{1}{k}\log \deg_{\pvec{x_{0}}}\varphi^{k},
        \label{eq:algentdef}
    \end{equation}
    \label{def:algent}
    called the \emph{algebraic entropy} is zero for every initial
    condition $\pvec{x_{0}}\in\Pj^{n}$.
\end{definition}

Algebraic entropy is an \emph{invariant} of bi-rational maps, meaning
that its value is unchanged up to bi-rational equivalence.
Practically algebraic entropy is a measure of the \emph{complexity} of
a map, analogous to the one introduced by Arnol'd \cite{Arnold1990}
for diffeomorphisms.
In this sense growth is given by computing the number of intersections of the
successive images of a straight line with a generic hyperplane in
complex projective space \cite{Veselov1992}.

The value of the degree of the iterates of the map is conditioned by its
\emph{singularity structure}. 
Some hypersurfaces are blown down by the map.
If one of the successive images of these hypersurfaces coincide with 
a singular variety, there is a drop in the degree 
\cite{BellonViallet1999,Viallet2015,Takenawa2001JPhyA}.
Therefore, from a heuristic point of view we can say that
the singularity makes the entropy. 
This actually also applies to non-autonomous cases
like the discrete Painlev\'e equations \cite{Sakai2001}.

In principle, the definition of algebraic entropy in equation
\eqref{eq:algentdef} requires us to compute all the iterates of
a bi-rational map $\varphi$ to obtain the sequence 
\begin{equation}
    d_{k}=\deg_{\pvec{x_{0}}}\varphi^{k},
    \quad
    k\in\N.
    \label{eq:degrees}
\end{equation}
Fortunately, for the majority of applications the form
of the sequence can be inferred by using generating
functions \cite{Lando2003}:
\begin{equation}
    g\left( z \right) = \sum_{n=0}^{\infty} d_{k}z^{k}.
\end{equation}
\label{eq:genfunc}
A generating function is a predictive tool which
can be used to test the successive members of a finite sequence.
It follows that the algebraic entropy is given by the
logarithm of the smallest pole of the generating function,
see \cite{GubbiottiASIDE16,GrammaticosHalburdRamaniViallet2009}.

Several results are known about the relationship of the
above definitions of integrability.
First of all, the low growth condition means that the complexity of
the map is very low, and it is known that invariants help in reducing
the complexity of a map. 
Indeed the growth of a map possessing invariants cannot be generic
since the motion is constrained to take place on the intersection 
of hypersurfaces defined by the invariants.
However, the drop in complexity must be big enough to reduce the
growth to a polynomial one.
On the other hand it is known that the existence of invariants
can give some bounds on the growth of bi-rational maps.
Indeed, it is known that the orbits of superintegrable maps with 
rational invariant are confined to \emph{elliptic curves} and the 
growth is at most \emph{quadratic} \cite{Bellon1999,Gizatullin1980}.
In low dimension some explicit results on the growth of bi-rational
maps are known.
For maps in $\Pj^{2}$, it was proved in \cite{DillerFavre2001} that 
the growth can be only bounded, linear, quadratic or exponential.
Linear cases are trivially integrable in the sense of invariants.
We note that for polynomial maps in $\Cp^{2}$, it was already known from 
\cite{Veselov1992} that the growth can be only linear or exponential.
It is known that QRT mappings and other maps with invariants in $\Pj^{2}$
possess quadratic growth \cite{Duistermaat2011book}, so the two notions
are actually equivalent for large class of integrable systems.

\subsection{Duality}

Now we  discuss briefly the concept of \emph{duality} for
rational maps, which was introduced in \cite{QuispelCapelRoberts2005}.
Let us assume that our map $\varphi$ possesses $L$ invariants, 
i.e.$I_{j}$ for $j\in\left\{ 1,\dots,L \right\}$. 
Then we can form the linear combination:
\begin{equation}
    H = \alpha_{1} I_{1} + \dots+\alpha_{L} I_{L}.
    \label{eq:hdef}
\end{equation}
%If $\deg_{x_{1}} H,\deg_{x_{n}} H >1$ we have that 
%
Being a function of invariants it follows that $H$ defined
by \eqref{eq:hdef} is itself an invariant of the map.

\begin{remark}
    We note that in principle more general combinations of
    invariants can be considered:
    \begin{equation}
        H = P_{d}\left(I_{1},I_{2},\dots,I_{L}\right)
        \label{eq:genHcomb}
    \end{equation}
    where $P_{d}$ is a homogeneous polynomial of total degree $d$
    in $L$ variables.
    Again even in this generalized case $H$ defined
    by \eqref{eq:genHcomb} is an invariant of the map.
    However, in this paper we won't consider this case, following
    the original definition of \cite{QuispelCapelRoberts2005}.
    \label{rem:gencomb}
\end{remark}

For an unspecified
recurrence relation
\begin{equation}
    \left[ x_{1}:x_{2}:\dots:x_{n+1} \right]
    \mapsto
    \left[ x_{1}':x_{2}':\dots:x_{n+1}' \right]
    = \left[ x_{1}':x_{1}:\dots:x_{n+1} \right]
\end{equation}
we can write down the invariant condition for $H$ \eqref{eq:hdef}:
\begin{equation}
    \widehat{H}(x_{1}',\pvec{x})=H\left(\pvec{x'}  \right) - H\left( \pvec{x} \right)
    = 0.
    \label{eq:hinvc}
\end{equation}
Since we know that $\pvec{x'}=\varphi\left( \pvec{x} \right)$ is a solution
of \eqref{eq:hinvc} we have the following factorization:
\begin{equation}
    \widehat{H}(x_{1}',\pvec{x})
    = A\left( x_{1}',\pvec{x} \right)
    B\left( x_{1}',\pvec{x} \right).
    \label{eq:hfact}
\end{equation}
%Since $H$ defined by \eqref{eq:hdef} is a invariant
%of $\varphi$ we  will have that one of the possible
%solutions of $H\left(\pvec{x'}  \right) - H\left( \pvec{x} \right)=0$
%is just given by $\pvec{x'}=\varphi\left( \pvec{x} \right)$.
We can assume without loss of generality that the map
$\varphi$ corresponds to the annihilation of $A$ in \eqref{eq:hfact}.
Now since $\deg_{x_{1}'}\widehat{H}=\deg_{x_{1}}H$ and 
$\deg_{x_{n}}\widehat{H}=\deg_{x_{n}}H$
we have that if if $\deg_{x_{1}} H,\deg_{x_{n}} H >1$ the factor
$B$ in \eqref{eq:hfact} is non constant\footnote{We remark that this
    assertion is possible because we are assuming that all the invariants
    are non-degenerate. It is easy to see that degenerate invariants can
violate this property.}.
In general, since the map $\varphi$ is bi-rational, we have the following
equalities:
\begin{subequations}
    \begin{align}
        \deg B_{x_{1}'}&=\deg_{x_{1}'}\widehat{H}-\deg_{x_{1}'}A=\deg_{x_{1}}H-1,
        \\
        \deg B_{x_{n}}&=\deg_{x_{n}}\widehat{H}-\deg_{x_{n}}A=\deg_{x_{n}}H-1.
    \end{align}
    \label{eq:degABH}
\end{subequations}
Therefore we have that, in general, if $\deg_{x_{1}}H,\deg_{x_{n}}H>2$, 
the annihilation of $B$ does not define a bi-rational map, but an algebraic one.
However when  $\deg_{x_{1}} H,\deg_{x_{n}} H =2$
the annihilation of $B$ defines a bi-rational projective map.
We call this map the \emph{dual map} and we denote it by $\varphi^{\vee}$.
%A dual map will be called \emph{self-dual} if $\varphi^{\vee}=\varphi$

\begin{remark}
    We note that in principle for $\deg_{x_{1}}H=\deg_{x_{n}}H=d>2$, 
    more general factorizations can be considered:
    \begin{equation}
        \widehat{H}\left( x_{1}',\pvec{x} \right)
        = \prod_{i=1}^{d} A_{i}\left( x_{1}',\pvec{x} \right),
        \label{eq:genHfact}
    \end{equation}
    but in this paper we won't consider this case.
    \label{rem:facthigh}
\end{remark}

Now assume that the invariants (and hence the map $\varphi$)
depends on some \emph{arbitrary constants} 
$I_{i}=I_{i}\left( \pvec{x};a_{i}\right)$, for $i=1,\dots,K$.
Choosing some of the $a_{i}$ in  such a way that
there remains $M$ arbitrary constants and such that for
a subset $a_{i_{k}}$ we can write
equation \eqref{eq:hdef} in the following way:
\begin{equation}
    H = a_{i_{1}} J_{1} + a_{i_{2}} J_{2}+\dots+a_{i_{K}}J_{a_{i_{K}}},
    \label{eq:hJdef}
\end{equation}
where $J_{i}=J_{i}\left( \pvec{x} \right)$, $i=1,2,\dots,K$
are new functions.
Then using the factorization \eqref{eq:hfact} we have that
the $J_{i}$ functions are invariants for the dual maps.

\begin{remark}
    It is clear from equation \eqref{eq:hJdef} that even though
    the dual map is naturally equipped with some invariants, 
    it is not \emph{necessarily} equipped with a sufficient
    number of invariants to claim integrability.
    In fact there exists examples of dual maps with any possible
    behaviour, integrable, superintegrable and non-integrable
    \cite{JoshiViallet2017,GJTV_sanya}.
    \label{eq:dualmapints}
\end{remark}

\section{Derivation of the class of 4D maps}
\label{sec:theproblem}

In this Section we explain how we derive the class of
4D maps with two invariants we are going to present in Section
\ref{sec:theclass}.

Our starting points are the maps corresponding to the autonomous 
\ref{dPI2} and the $\dP_\text{II}^{2}$ equations and their invariants
as presented in \cite{JoshiViallet2017}.
These two are maps of $\Pj^{4}$ into itself with
coordinates $\left[ x:y:z:u:t \right]$.
Their components are given by:
\begin{equation}
   \begin{aligned}
        x'  &=  
        \begin{aligned}[t]
        &- a y \left({x}^{2}+ {y}^{2}+{z}^{2}  +2 yz+2 xy
        +xz+zu \right) 
        \\
        &\quad
        - b  t y  \left( y+z+x \right)  -c   y{t}^{2}+ d  {t}^{3},
        \end{aligned}
        \\
        y' &=  a  y{x}^{2}, \quad z'  = a x {y}^{2}, \quad
        u' = a  xyz, \quad
        t'  = a  xyt .
\end{aligned}
    \tag{$\dP_\text{I}^{2}$}
    \label{dPI2}
\end{equation}
and by:
\begin{equation}
\begin{aligned}
 x'  &=  d  {t}^{5} - a   \left( t-y \right)  \left( t+y \right)  \left( u{t}^{2}-y{z}^{2}-u{z}^{2}-2 yxz-{x}^{2}y \right)  -c  y {t}^{4}  \\
 &\quad 
  - b   {t}^{2} \left( t-y \right)  \left( t+y \right)  \left( z+x \right), 
\\
 y' &= ax \left( {t}^{2}-{y}^{2} \right)  \left( {t}^{2}-{x}^{2} \right) , 
\quad z'  = ay \left( {t}^{2}-{y}^{2} \right)  \left( {t}^{2}-{x}^{2} \right) , 
\\ 
 u' &= az \left( {t}^{2}-{y}^{2} \right)  \left( {t}^{2}-{x}^{2} \right) , 
\quad
  t'  = at \left( {t}^{2}-{y}^{2} \right)  \left( {t}^{2}-{x}^{2} \right).
\end{aligned}
    \tag{$\dP_\text{II}^{2}$}
    \label{dPII2}
\end{equation}

It can be checked that the map \ref{dPI2} has two invariants 
$I^{\scriptscriptstyle (\rm I) }_4$  and $I^{\scriptscriptstyle (\rm I) }_5$ 
which are:
\begin{subequations}
    \begin{align}
        \label{II4}
        t^4 I^{\scriptscriptstyle (\rm I) }_4 &
        \begin{aligned}[t]
        &= a   
        yz \left( -{y}^{2}-2\,yz-xy-{z}^{2}-zu+xu \right)  
        \\
        &\quad
        -  b  tyz \left( z+y \right)
        -c   yz{t}^{2}+  d {t}^{3} \left( z+y \right),
        \end{aligned}
        \\
        t^5 I^{\scriptscriptstyle (I)}_5&  
        \begin{aligned}[t]
        &=a  y z \left( zu+xy+{y}^{2}+2\,yz+{z}^{2} \right)  \left( z+u+y+x \right)
        \\ 
        &
        \quad
        +c  yz \left( z+u+y +x \right) {t}^{2}  
        -d \left( zu+xy+{y}^{2}+2\,yz+{z}^{2} \right) {t}^{3}
        \\
        &\quad
        +b  yz \left( y+z+x \right)  
        \left( u+y+z \right) t, 
        \end{aligned}
        \label{II5}
    \end{align}
    \label{eq:dPI2inv}
\end{subequations}
while the map \ref{dPII2} possesses two invariants 
$I^{\scriptscriptstyle (\rm II) }_6$  and $I^{\scriptscriptstyle (\rm II) }_8$ 
given by
\begin{subequations}
    \begin{align}
        \label{III6}
        t^6 I^{\scriptscriptstyle (\rm II) }_6& 
        \begin{aligned}[t]
            &=   a  \left( t-z \right)  \left( t+z \right)  \left( t-y \right)  
            \left( t+y \right)  \left( ux-uz-xy-yz \right)  
            \\ 
            & \quad 
            -b  {t}^{2} \left( {z}^{2}{t}^{2}+{t}^{2}{y}^{2}-{z}^{2}{y}^{2} \right)  
            -c  {t}^{4}yz+ d   {t}^{5} \left( z+y \right),
        \end{aligned}
        \\
        \label{III8}
        t^8 I^{\scriptscriptstyle (\rm II) }_8&  
        \begin{aligned}[t]
            &= a   \bigl[ 
             \left( {u}^{2}+{z}^{2}+{y}^{2}+{x}^{2} \right) {t}^{6} -{z}^{2}{y}^{2} \left( uz+xy+yz \right) ^{2}
            \\ 
            &\quad
             - ( 2 yu{z}^{2}+2 uzxy+{x}^{2}{z}^{2}+2 xz{y}^{2}+2 {x}^{2}{y}^{2}
            \\
            &\quad+{u}^{2}{y}^{2}+2 {z}^{2}{y}^{2}+2 {u}^{2}{z}^{2}) {t}^{4} 
            \\ 
            &\quad
            + ( 2 {x}^{2}{y}^{2}{z}^{2}+2 u{y}^{3}{z}^{2}+2 x{y}^{2}{z}^{3}+2 yu{z}^{4}
            \\
            &\quad\quad
            +{z}^{2}{y}^{4}+{y}^{2}{z}^{4}+2 {u}^{2}{y}^{2}{z}^{2}
            +{x}^{2}{y}^{4} 
            \\ 
            &\quad\quad
            +2 ux{y}^{3}z+2 uxy{z}^{3}+2 xz{y}^{4}+{z}^{4}{u}^{2} ) {t}^{2}\bigr]
            \\ 
            &\quad
            + b   {t}^{2} \left( t-z \right)  \left( t+z \right)  \left( t-y \right)  
            \left( t+y \right)  \left( z+x \right)  \left( u+y \right) 
            \\ 
            &\quad
            + c  {t}^{4} \left( xz{t}^{2}-{z}^{2}{y}^{2}+yu{t}^{2}-yu{z}^{2}
            -xz{y}^{2} \right)  
            \\ 
            &
            \quad
            - d  {t}^{5} \left( x{t}^{2}+z{t}^{2}-z{y}^{2}
            -x{y}^{2}-u{z}^{2}+u{t}^{2}-y{z}^{2}+y{t}^{2} \right).
        \end{aligned}
    \end{align}
    \label{eq:dPII2inv}
\end{subequations}

The invariants of the maps \ref{dPI2} and \ref{dPII2} have
the following properties:
\begin{enumerate}[label={Property \Alph*:},font={\bfseries},leftmargin=78pt]%[Property A:]
    \item The invariants are symmetric in the sense of definition
        \ref{def:symm}.
    \item The lowest order invariants \eqref{II4} and \eqref{III6} 
        have degree pattern $\left( 1,3,3,1 \right)$ and are
        particular instances of the homogeneous polynomial in
        $\Cp\left[ x,y,z,u,t \right]$:
        \begin{equation}
            \begin{aligned}
                t^{6}I_{\text{low}} &= 
                t^5 (y+z) s_{1}-t^4 (u x-u z-x y) s_{2}+s_{3} t^4 y z
                \\
                &\quad
                +t^4 (y^2+z^2) s_{4}+t^3 y z (y+z) s_{5}
                \\
                &\quad
                +t^2 (y^2+z^2) (u x-u z-x y) s_{6}
                \\
                &
                \quad
                -t^2 y z (u x-u z-x y) s_{7}+s_{8} t^2 y^2 z^2+t^2 y z (y^2+z^2) s_{9}
                \\
                &\quad
                -y^2 z^2 (u x-u z-x y) s_{10}+s_{11} y^3 z^3,
            \end{aligned}
            \label{eq:Ilow}
        \end{equation}
        depending parametrically on 11 coefficients, namely 
        $s_{i}$, $i=1,\dots,11$. 
    \item The highest order invariants \eqref{II5} and \eqref{III8} 
        have degree pattern $\left( 2,4,4,2 \right)$.
        The most general homogeneous polynomial in
        $\Cp\left[ x,y,z,u,t \right]$ depends parametrically on 1820
        coefficients. Taking into account the symmetry with respect to
        the involution \eqref{eq:symmcond} the number of coefficients 121.
        Since, one of this coefficients is just an additional constant then
        we can lower the number of independent coefficient to 120.
        We denote this invariant by $\Ihigh$, but 
        we do not present the general form of this
        polynomial here, since it will be too cumbersome to write down.
\end{enumerate}

Based on the above consideration it is natural to address the following
problem:

\begin{problem}
    Find all the bi-rational maps $\varphi\colon\Pj^{4}\to\Pj^{4}$
    and their dual maps $\varphi^{\vee}\colon\Pj^{4}\to\Pj^{4}$ 
    having two non-degenerate, 
    functionally independent invariants with properties A, B and C.
    \label{prob:abc}
\end{problem}

Solving this problem amount to obtain a list of equations which
are expected to behave like the two fourth-order Painlev\'e equations
\ref{dPI2} and \ref{dPII2}.
Before going to the solution of this problem, let us remark the
following general result on the dual map of a map with two invariants
possessing properties A, B and C:

\begin{lemma}
    Assume that a map $\varphi\colon\Pj^{4}\to\Pj^{4}$ possesses two
    invariants with properties A, B and C.
    Then we have the map $\varphi$ has degree pattern 
    $\dpat \varphi= \left( 2,3,2,1 \right)$ 
    and the maximal degree pattern of  the dual  
    $\varphi^{\vee}\colon\Pj^{4}\to\Pj^{4}$ is 
    $\dpat \varphi^{\vee}= \left( 2,1,2,1 \right)$.
    \label{lem:invstruct}
\end{lemma}

\begin{proof}
    By direct computation is it possible to check that if an invariant
    $\Ilow$ has the form \eqref{eq:Ilow} then the invariant
    condition \eqref{eq:fintdef} implies the following factorisation:
    \begin{equation}
        \Ilow\left(\pvec{x'}  \right) - \Ilow\left( \pvec{x} \right)
        = (x-z) A\left(x', \pvec{x} \right).
        \label{eq:I1fact}
    \end{equation}
    Equation \eqref{eq:I1fact} means that we have the
    following degree distribution:
    \begin{equation}
        \begin{array}{cccccc}
            & \deg_{x'} & \deg_{x} & \deg_{y} & \deg_z & \deg_u
            \\
            \varphi^{*}\left( \Ilow \right) & 1 & 3 & 3 & 1 & 0
            \\
            \Ilow & 0 & 1 & 3 & 3 & 1
            \\
            A & 1 & 2 & 3 & 2 & 1 
        \end{array}
        \label{eq:dconfI1}
    \end{equation}
    The second part of the statement comes from an analogous consideration 
    applied to equation \eqref{eq:hfact}. 
    Since the degree pattern of $A$ is fixed, the degree pattern 
    of $B$ is maximal when there are no factors depending only on $\pvec{x}$. 
    Under this assumptions we find the following distribution of the degrees:
    \begin{equation}
        \begin{array}{cccccc}
            & \deg_{x'} & \deg_{x} & \deg_{y} & \deg_z & \deg_u
            \\
            \varphi^{*}\left( H \right) & 2 & 4 & 4 & 2 & 0
            \\
            H & 0 & 2 & 4 & 4 & 2
            \\
            A & 1 & 2 & 3 & 2 & 1 
            \\
            B & 1 & 2 & 1 & 2 & 1 
        \end{array}
        \label{eq:dconfH}
    \end{equation}
    This ends the proof.
\end{proof}

\begin{corollary}
    Bi-rational maps possessing two invariants satisfying
    properties A, B, and C in general are not \emph{self-dual}.
    \label{cor:noself}
\end{corollary}

We sketch now the procedure we used to solve problem \ref{prob:abc}.
We underline that this procedure is based on the one proposed in 
\cite{JoshiViallet2017} to find bi-rational maps with invariants
of assigned degree pattern.

\begin{enumerate}
    \item Find the value of $x'$ from \eqref{eq:I1fact}
        where $I_\text{low}$ is given by equation \eqref{eq:Ilow}.
    \item Substitute the obtained form of $x'$ into the invariant
        condition \eqref{eq:fintdef} for $I_\text{high}$.
        Geometrically this describes the intersection of the two 
        hypersurfaces given by $I_{\text{low}}=I_{\text{low}}^{(0)}$ and
        $I_\text{high}=I_\text{hight}^{(0)}$, where
        $I_\text{low}^{(0)}$ and $I_\text{hight}^{(0)}$ are arbitrary constants.
    \item We can take coefficients with respect to the independent 
        variables.
        This yield a system of nonlinear homogeneous equations.
        We put this system in a collection of systems that we call.
    \item We convert this system to a set of simpler systems
        by solving iteratively all the monomial equations of each 
        system.
        At each stage we exclude the systems originating invariants
        contraddicting properties A, B and C.
    \item This yields 117 different smaller systems.
    \item Solving these systems we found 25 solutions respecting
        the properties A, B and C.
\end{enumerate}

Through a degeneration scheme the 25 solutions we obtain 
can be cast into six different maps along with their duals. 
We proved the following:

\begin{theorem}
    The solutions of problem \ref{prob:abc}, up to degeneration and 
    identification of the free parameters, is given by \emph{six pairs}
    of main/dual maps which we denote by (P.x) with x small roman
    number for the main maps and by (Q.x) for the dual map.
    \label{thm:class}
\end{theorem}

We call this class of maps the \emph{(P,Q) class}.
In the next section we present the explicit form of these maps and
we discuss their integrability properties.

\section{Maps of the (P,Q) class and their integrability properties}
\label{sec:theclass}

In this section we show the explicit form of the maps of the
class (P,Q).
We denote the pairs of main/dual maps by (x) where x is a small
roman number.
Moreover we discuss their integrability properties from
the point of view of the existence of invariants, 
the degree growth of their iterates
and the existence of Lagrangians. For the cases admitting Lagrangian
following corollary \ref{cor:costrpb} we present the form of their
symplectic structure.

\stepcounter{eqlistred}
\subsection{Maps (i)}

%The maps ired\ are the same as the ones of list VII.
The main map $\pvec{x}\mapsto\varphi_\text{i}\left( \pvec{x} \right)=\pvec{x'}$ 
has the following components:
\begin{equation}
    \begin{aligned}
        x'&
        \begin{aligned}[t]
            &=
            -\{[\nu t^2 (x+ z)+u z^2] y+t^2 \mu u z+(x+z)^2 y^2\} d-a t^4,
        \end{aligned}
        \\
        y'&=x^2 d (t^2 \mu+x y), 
        \quad
        z'=y x d (t^2 \mu+x y),
        \\
        u'&= z x d (t^2 \mu+x y), 
        \quad
        t'= t x d (t^2 \mu+x y).
    \end{aligned}
    \tag{P.i}
    \label{eq:Mi}
\end{equation}
This map depends on four parameters $a, d$ and $\mu,\nu$.
The map  \eqref{eq:Mi} has the following degrees of iterates:
\begin{equation}
    \begin{aligned}
        \left\{d_n\right\}_{\text{\ref{eq:Mi}}}
        &= 1, 4, 12, 28, 52, 86, 130, 188, 260, 348, 452,  
        \\
        &\phantom{+}576, 720,886, 1074, 1288, 1528, 1796, 2092\dots
    \end{aligned}
    \label{eq:degMi}
\end{equation}
with generating function:
\begin{equation}
    g_{\text{\ref{eq:Mi}}}(s) 
    =\frac{s^7-3 s^6+s^5-s^4+3 s^3+3 s^2+s+1}{(s+1)(s^2+1)(s-1)^4}.
    \label{eq:gfMi}
\end{equation}
All the poles of the generating function \eqref{eq:gfMi} lie
on the unit circle, so that the map \eqref{eq:Mi} is integrable
according to the algebraic entropy criterion.
Moreover, due to the presence of $\left( s-1 \right)^{4}$ in the denominator
of the generating function \eqref{eq:gfMi} we have that the main map
\eqref{eq:Mi} has \emph{cubic} growth.

The map \eqref{eq:Mi} has the following invariants:
\begin{subequations}
    \begin{align}
        t^{6} \Ilow^{\text{\ref{eq:Mi}}} &
        \begin{aligned}[t]
        &=
        a t^4 y z+d\left[ \nu y^2 z^2-y z  (u x-u z-x y) \mu\right] t^2
        \\
        & \quad
        -y^2 z^2 d (u x-x y-y z-u z),
        \end{aligned}
        \label{eq:IlowMi}
        \\
        t^{8}\Ihigh^{\text{\ref{eq:Mi}}} &
        \begin{aligned}[t]
        &=
        \left[(u z+x y-y z) \mu-\nu y  z\right] a t^6
        \\
        &\quad
        +\left[y z (x y+y z+u z) a+d (u z+x y-y z)^2 \mu^2\right.
        \\
        &\quad
        \phantom{+}\left.+2 d y z (u x-y z)  \mu \nu -d \nu^2 y^2 z^2\right] t^4
        \\
        &\quad
        +\left[2 d z y (u z+x y-y z) (x y+y z+u z) \mu+2 d y^2 z^2 \nu u x\right] t^2
        \\
        &\quad
        +y^2 z^2 d (x y+y z+u z)^2.
        \end{aligned}
        \label{eq:IhighMi}
    \end{align}
    \label{eq:intMi}
\end{subequations}

The two invariants \eqref{eq:intMi} alone cannot explain why 
the map \eqref{eq:Mi} is integrable according to the algebraic entropy
criterion.
Indeed, as we stressed in section \ref{sec:setting} two invariants
are not enough to claim Liouville integrability, nor to claim integrability
in the sense of definition \ref{def:invint}.
We can show using the method presented in \cite{Gubbiotti_dcov} that the
map \eqref{eq:Mi} is not variational.
However, as we shown in \cite{GJTV_sanya}, the recurrence associated to the
map \eqref{eq:Mi} can be deflated to a three dimensional map via the transformation
$v_{k}=w_{k}w_{k+1}$. 
The map obtained in this way is then integrable in the sense of Liouville.
For all the details we refer to \cite{GJTV_sanya}.

The dual map $\pvec{x}\mapsto\varphi^{\vee}_\text{i}\left( \pvec{x} \right)=\pvec{x'}$ 
has the following components:
\begin{equation}
    \begin{aligned}
        x'&
        \begin{aligned}[t]
            &=
            \left[\beta (2 x y-2 y z+u z) \mu+(\beta \nu-\alpha) y (x-z)\right] t^2
            \\
            &\quad
            +\beta y (z^2 y-x^2 y+u z^2)
        \end{aligned}
        \\
        y' &=  x^2 \beta  (t^2 \mu+x y), \quad
        z' =  y x \beta  (t^2 \mu+x y), 
        \\
        u' &= z x \beta  (t^2 \mu+x y), \quad
        t' = t x \beta  (t^2 \mu+x y).
    \end{aligned}
    \tag{Q.i}
    \label{eq:Di}
\end{equation}
This map depends on three parameters $\alpha,\beta$, and $\mu,\nu$.
The parameters $\mu$ and $\nu$ are shared with the main map \eqref{eq:Mi}.
The map \eqref{eq:Di} has the following degrees of iterates:
\begin{equation}
    \begin{aligned}
        \left\{d_n\right\}_{\text{\ref{eq:Di}}}
        &= 1, 4, 12, 26, 48, 78, 118, 170, 234, 312, 406, 516,644,792\dots
    \end{aligned}
    \label{eq:degDi}
\end{equation}
with generating function:
\begin{equation}
    g_{\text{\ref{eq:Di}}}(s) =
    \frac{(s^3-2s^2-1)(s^3-s^2-s-1)}{(s^2+s+1)(s-1)^4}.
    \label{eq:gfDi}
\end{equation}
This means that the dual map is integrable according to the
algebraic entropy test with \emph{cubic} growth, just like the main map.

The main map \eqref{eq:Mi} possesses two invariants and 
depends on $a$ and $d$ whereas the dual map \eqref{eq:Di} do not
depend on them. Then according to \eqref{eq:hJdef} we can write down the first invariants
for the dual map \eqref{eq:Di} as:
\begin{equation}
    \alpha\Ilow^{\text{\ref{eq:Mi}}}+\beta\Ihigh^{\text{\ref{eq:Mi}}}
    =
    a\Ilow^{\text{\ref{eq:Di}}} +d \Ihigh^{\text{\ref{eq:Di}}}.
    \label{eq:derDi}
\end{equation}
Therefore we obtain the following expressions:
\begin{subequations}
    \begin{align}
        t^{4}\Ilow^{\text{\ref{eq:Di}}} &
        \begin{aligned}[t]
            &=
            (y z \alpha+(\mu x y-y z \mu-y \nu z+\mu u z) \beta) t^2
            \\
            &\quad
            +\beta y z (x y+y z+u z),
        \end{aligned}
        \label{eq:I1i}
        \\
        t^{8}\Ihigh^{\text{\ref{eq:Di}}} &
        \begin{aligned}[t]
            &=
            \left\{\left[y^2 z^2 \nu-y z (u x-u z-x y) \mu\right] \alpha\right.
            \\
            &\quad
            \phantom{+}\left.+\left[(u z+x y-y z)^2 \mu^2+2 y z (u x-y z) \nu \mu-\nu^2 y^2 z^2\right] \beta\right\} t^4
            \\
            &\quad
            +\left\{z^2 y^2 (x y+y z-u x+u z) \alpha\right.
            \\
            &\quad
            \phantom{+}\left.+\left[2 y z (u z+x y-y z) (x y+y z+u z) \mu+2 y^2 z^2 \nu u x\right] \beta\right\} t^2
            \\
            &\quad
            +z^2 y^2 (x y+y z+u z)^2 \beta.
        \end{aligned}
        \label{eq:I2i}
    \end{align}
    \label{eq:intDi}
\end{subequations}
The invariant \eqref{eq:I1i}  
has degree pattern $\left( 1,2,2,1 \right)$.

The properties of the dual map \eqref{eq:Di} are very similar
to those of the main map \eqref{eq:Mi}.
Again, the two invariants \eqref{eq:intDi} alone cannot explain the 
low growth and following \cite{Gubbiotti_dcov} no Lagrangian exists.
However, as was shown in \cite{GJTV_sanya}, the recurrence associated to the
map \eqref{eq:Di} can be deflated to a three dimensional via the transformation
$v_{k}=w_{k}w_{k+1}$. 
Again, the map obtained in this way is integrable in the sense of Liouville.
For all the details we refer again to \cite{GJTV_sanya}.

\subsection{Maps (ii)}

%\paragraph{Main map.}
The main map $\pvec{x}\mapsto\varphi_\text{ii}\left( \pvec{x} \right)=\pvec{x'}$ 
has the following components:
\begin{equation}
    \begin{aligned}
        x' &=\left[(x^2+z^2) y-u z^2\right] \mu-t^2 (u-2y), 
        \\
        y' &= x (t^2+\mu x^2), 
        \quad
        z' = y (t^2+\mu x^2),
        \\
        u' &= z (t^2+\mu x^2), 
        \quad
        t'= t (t^2+\mu x^2).
    \end{aligned}
    \tag{P.ii}
    \label{eq:Pii}
\end{equation}
This map depends on the parameter $\mu$.
The map \eqref{eq:Pii} has the following degrees of iterates:
\begin{equation}
    \begin{aligned}
        \left\{d_n\right\}_{\text{\ref{eq:Pii}}}
        &=1,3,9,21,45,93,189,381,765,1533\dots
    \end{aligned}
    \label{eq:degP\theeqlist}
\end{equation}
with generating function:
\begin{equation}
    g_{\text{\ref{eq:Pii}}}(s) = \frac{1+2 s^2}{(2 s-1)(s-1)}.
    \label{eq:gfPii}
\end{equation}
This means that the main map is non-integrable according to the
algebraic entropy test with positive entropy $\varepsilon=\log2$.

Despite being non-integrable the main map 
\eqref{eq:Pii} has, by construction, the following invariants:
\begin{subequations}
    \begin{align}
        t^{4}\Ilow^{\text{\ref{eq:Pii}}} &
        \begin{aligned}[t]
        &=\left( x-z \right)\left( u-y \right)
        \left( {t}^{2}+{z}^{2}\mu \right)  \left( \mu{y}^{2}+{t}^{2} \right)     
        \end{aligned}
        \label{eq:IlowPii}
        \\
        t^{6}\Ihigh^{\text{\ref{eq:Pii}}} &
        \begin{aligned}[t]
        &=
        \left[  \left( x-z \right) ^{2}{y}^{4}+{y}^{2}{z}^{4}-2 y{z}^{4}u+{u}^{2}{z}^{4} \right] {\mu}^{2}
        \\
        &\quad
        +2 {t}^{2} \left[  \left( {x}^{2}-2 xz+2 {z}^{2} \right) {y}^{2}-2 y{z}^{2}u+{u}^{2}{z}^{2} \right] \mu
        \\
        &\quad
        +{t}^{4} \left( {z}^{2}+{u}^{2}+{x}^{2}+{y}^{2}-2 uy-2 xz \right)
        \end{aligned}
        \label{eq:IhighPii}
    \end{align}
    \label{eq:intMii}
\end{subequations}

Moreover, using the test of \cite{Gubbiotti_dcov} we have that
the map \eqref{eq:Pii} is not variational.

%\paragraph{Dual map.}
The dual map $\pvec{x}\mapsto\varphi^{\vee}_\text{ii}\left( \pvec{x} \right)=\pvec{x'}$ 
is given by the following components:
\begin{equation}
    \begin{aligned}
        x'&
        \begin{aligned}[t]
        &= 
        \alpha  \left[  \left( {x}^{2}-{z}^{2} \right) y+u{z}^{2} \right] {\mu}^{2}
        +  {t}^{2}\alpha u+\beta {y}^{2} \left( x-z \right)   \mu
        +{t}^{2}\beta  \left( x-z \right), 
        \end{aligned}
        \\
        y' &= \alpha x  \left( {t}^{2}+\mu{x}^{2} \right),
        \quad
        z' = \alpha y  \left( {t}^{2}+\mu{x}^{2} \right), 
        \\
        u' &= \alpha z  \left( {t}^{2}+\mu{x}^{2} \right),
        \quad
        t' = \alpha t  \left( {t}^{2}+\mu{x}^{2} \right).
    \end{aligned}
    \tag{Q.ii}
    \label{eq:Qii}
\end{equation}
This map depends on three parameters $\alpha,\beta$ and $\mu$.
The parameter $\mu$ is shared with the main map \eqref{eq:Pii}.
The map given by \eqref{eq:Qii} has the following degrees of iterates:
\begin{equation}
    \begin{aligned}
        \left\{d_n\right\}_{\text{\ref{eq:Qii}}}
        &=1, 3, 9, 21, 45, 93, 189, 381, 765, 1533, 3069\dots
    \end{aligned}
    \label{eq:degQii}
\end{equation}
with generating function:
\begin{equation}
    g_{\text{\ref{eq:Qii}}}(s) = \frac{1+2 s^2}{(2 s-1)(s-1)}.
    \label{eq:gfQii}
\end{equation}
This means that the main map is non-integrable according to the
algebraic entropy test with positive entropy $\varepsilon=\log2$.
We remark that the growth is exactly the same as the main map.

Since the main map \eqref{eq:Pii} possesses two invariants,
but it has only one parameter $\mu$ shared with the dual map.
Then according to \eqref{eq:hJdef} we can only write down a single invariant 
for the dual map \eqref{eq:Qii} as:
\begin{equation}
    I^{\text{\ref{eq:Qii}}}=
    \alpha\Ihigh^{\text{\ref{eq:Pii}}}+\beta\Ilow^{\text{\ref{eq:Pii}}}.
    \label{eq:derQii}
\end{equation}
The invariant \eqref{eq:derQii} has degree pattern
$\left( 2,4,4,2 \right)$.

Finally, using the test of \cite{Gubbiotti_dcov} we have that
the map \eqref{eq:Qii} is not variational, as the main map \eqref{eq:Pii}.

For additional comments about the maps \eqref{eq:Pii} and \eqref{eq:Qii}
we refer to \cite{GJTV_sanya}.

\subsection{Maps (iii)}

The main map $\pvec{x}\mapsto\varphi_\text{iii}\left( \pvec{x} \right)=\pvec{x'}$ 
is has the following components:
\begin{equation}
    \begin{aligned}
        x'&
        \begin{aligned}[t]
            &=-2 a t^4-2 \mu d (z+x+y) t^3 
            \\
            &
            \quad
            +\nu d \left[2(x+z) y+u z\right] t^2+d(2 y z^2 u+2 y^2 z x),
        \end{aligned}
        \\
        y'&= x^2 d (\nu t^2+2 x y),
        \quad
        z'=y x d (\nu t^2+2 x y),
        \\
        u'&=z x d (\nu t^2+2 x y), 
        \quad
        t'=t x d (\nu t^2+2 x y).
    \end{aligned}
    \tag{P.iii}
    \label{eq:Miii}
\end{equation}
This map depends on four parameters $a,d$ and $\mu,\nu$.
The map \eqref{eq:Miii} has the following degrees of iterates:
\begin{equation}
    \begin{aligned}
        \left\{d_n\right\}_{\text{\ref{eq:Miii}}}
        &= 1, 4, 12, 26, 49, 79, 113, 153, 199, 250, 310, 378, 449,
        \\
        &\phantom{=}526, 610, 698, 795,    901, 1009, 1123, 1245\dots
    \end{aligned}
    \label{eq:degMiii}
\end{equation}
with generating function:
\begin{equation}
    g_{\text{\ref{eq:Miii}}}(s) 
    =-\frac{4 s^8+4 s^7+10 s^6+9 s^5+13 s^4+7 s^3+6 s^2+2 s+1}{(s^2-s+1)(s^2+s+1)^2(s-1)^3}.
    \label{eq:gfMiii}
\end{equation}
This means that the main map is integrable according to the
algebraic entropy test with \emph{quadratic} growth.

The main map \eqref{eq:Miii} has the following invariants:
\begin{subequations}
    \begin{align}
        t^{6}\Ilow^{\text{\ref{eq:Miii}}} &
        \begin{aligned}[t]
        &=
        2 a t^4 y z+2 y z \mu d (y+z) t^3-y z d \nu (-2 y z-x y-u z+u x) t^2
        \\
        &\quad
        -2 y^2 z^2 (u x-u z-x y) d,
        \end{aligned}
        \label{eq:IlowMiii}
        \\
        t^{8}\Ihigh^{\text{\ref{eq:Miii}}} &
        \begin{aligned}[t]
        &=
        4 \mu a (y+z) t^7
        \\
        &\quad
        +(4 d y z \mu^2-2 a y z \nu+4 d z^2 \mu^2+4 d y^2 \mu^2+2 a z \nu u+2 a y \nu x) t^6
        \\
        &\quad
        +2 \mu \nu d (2 u z^2+2 x y^2+y z x+z u y) t^5
        \\
        &\quad
        +(d \nu^2 x^2 y^2-3 d \nu^2 y^2 z^2+4 d \nu^2 u x y z
        \\
        &\quad
        \phantom{+}+d \nu^2 u^2 z^2+4 a y^2 z x+4 a y z^2 u) t^4
        \\
        &\quad
        +4 \mu d y z (2 u z^2+2 x y^2+y z x+z u y) t^3
        \\
        &\quad
        +2 y d \nu z (u z+2 x y) (2 u z+x y) t^2
        \\
        &\quad
        +4 y^2 z^2 d (u^2 z^2+x^2 y^2+u x y z).
        \end{aligned}
        \label{eq:IhighMiii}
    \end{align}
    \label{eq:intMiii}
\end{subequations}

Moreover, we note that according to the test in \cite{Gubbiotti_dcov} the
main map \eqref{eq:Miii} does not posses a Lagrangian.
However, by direct search, we can prove that this map has an additional 
functionally independent invariant of degree pattern $(2,5,5,2)$. 
This means that the low growth of the main map \eqref{eq:Miii} is
explained in terms of integrability as existence of invariants, as given 
by definition \ref{def:invint}. 
More specifically, the quadratic growth it explained by the fact that
if a map in $\Pj^{4}$ has three rational invariants, the orbits are
confined to elliptic curves and the growth is at most quadratic \cite{Bellon1999}.

%\paragraph{Dual map.}
The dual map $\pvec{x}\mapsto\varphi^{\vee}_\text{iii}\left( \pvec{x} \right)=\pvec{x'}$ 
has the following components:
\begin{equation}
    \begin{aligned}
        x'&
        \begin{aligned}[t]
            &=
            2 \mu \beta (z-x) t^3
         + \left\{\beta \nu [zu + 2y\left( x-z \right) ] + 2 \alpha y ( z- x)\right\} t^2
            +2 \beta z^2 y u
        \end{aligned}
        \\
        y' &= x^2 \beta  (\nu t^2+2 x y),
        \quad
        z' = y x \beta  (\nu t^2+2 x y), 
        \\
        u' &= z x \beta  (\nu t^2+2 x y),
        \quad
        t' = t x \beta  (\nu t^2+2 x y).
    \end{aligned}
    \tag{Q.iii}
    \label{eq:Diii}
\end{equation}
This map depends on four parameters $\alpha,\beta$ and $\mu.\nu$.
The parameters $\mu,\nu$ are shared with the main map \eqref{eq:Miii}.
The map \eqref{eq:Diii} has the following degrees of iterates:
\begin{equation}
    \begin{aligned}
        \left\{d_n\right\}_{\text{\ref{eq:Diii}}}
        &=1,4,12,28,62,131,272,554,1120,2253,4528, 9092, 
        \\
        &\phantom{=}18244,36601, 73420, 147270, 295392, 592487, 
        \\
        &\phantom{=}1188378, 2383576, 4780824, 9589061,19233098,
        \\
        &\phantom{=} 38576452, 77374040, 155191611, 311272822, 624329930\dots,
    \end{aligned}
    \label{eq:degDiii}
\end{equation}
%with generating function:
%\begin{equation}
%    g_{\text{\ref{eq:Diii}}}(s) = ?.
%    \label{eq:gfDiii}
%\end{equation}
with generating function:
\begin{equation}
    g_{\text{\ref{eq:Diii}}}(s) =
    \frac{P_{\text{\ref{eq:Diii}}}(s)}{Q_{\text{\ref{eq:Diii}}}(s)},
    \label{eq:gfDiii}
\end{equation}
where
\begin{subequations}
    \begin{align}
    P_{\text{\ref{eq:Diii}}}(s)&
    \begin{aligned}[t]
        &=
            s^{14} + 2 s^{13} + 4 s^{12} + 6 s^{10} - s^{9} + 5 s^{8} 
            \\
            &+s^{7}+ 5 s^{6} + s^{5} + 4 s^{4} + s^{3} + 3 s^{2} + s + 1
    \end{aligned}
        \label{eq:PQiii}
    \\
        Q_{\text{\ref{eq:Diii}}}(s) &=\left(1-s\right) \left(s^{2} + 1\right) \left(s^{10} - 2 s^{9} - 2 s + 1\right). 
        \label{eq:QQiii}
    \end{align}
    \label{eq:PQQiii}
\end{subequations}
The growth of the map \eqref{eq:Diii} is given by the inverse of
the smallest pole of the generating function \eqref{eq:gfDiii}.
These poles are given by the zeroes of the function 
$Q_{\text{\ref{eq:Diii}}}(s)$ in \eqref{eq:QQiii}.
Clearly the zeroes of $1-s$ and $s^{2}+1$, lie on the unit circle,
therefore we have to look at the location of the zeroes of
the polynomial:
\begin{equation}
    q\left( s \right) = s^{10} - 2 s^{9} - 2 s + 1.
    \label{eq:qdef}
\end{equation}
Defining $q_{1}\left( s \right) = -2s+1$ and $q_{2}\left( s \right)=s^{10}-s^{9}$
we have that on the circle $C_{\rho}:=\left\{ s\in\Cp \middle| \abs{s} = \rho \right\}$
with $\rho\in \left( 1/2,1 \right)$ the following inequality holds:
\begin{equation}
    \abs{q_{2}\left(  s\right)} < \abs{q_{1}\left( s \right)}.
    \label{eq:rouche}
\end{equation}
By Rouche's theorem \cite{Ablowitz2003} this implies that $q_{1}\left( s \right)$ and
$q_{1}\left( s \right)+q_{2}\left( s \right)=q\left( s \right)$ have the
same number of zeroes inside the circle $C_{\rho}$, i.e. the polynomial
$Q_{\text{\ref{eq:Diii}}}(s)$ has a unique zero inside the circle $C_{\rho}$.
This zero is the smallest one of  $Q_{\text{\ref{eq:Diii}}}(s)$ and due to the fact 
that $Q_{\text{\ref{eq:Diii}}}(s)$ has real coefficients this zero is real.
This implies the growth of the dual map \eqref{eq:Diii} is \emph{exponential}.
The approximate value of the zero of $Q_{\text{\ref{eq:Diii}}}(s)$ inside $C_{\rho}$
is $s_{0}=0.49857104591719819\dots$. This implies that the algebraic
entropy of the dual map \eqref{eq:Diii} is:
\begin{equation}
    \varepsilon_\text{\ref{eq:Diii}} = \log \left(2.0057321984279013\dots\right).
    \label{eq:aeQiii}
\end{equation}
The growth of the sequence of degrees of equation \eqref{eq:Diii} is then 
slightly greater than $2^{n}$.

Since the main map \eqref{eq:Miii} possesses two invariants and 
depends on $a$ and $d$ whereas the dual map \eqref{eq:Diii} do not
depend on them according to \eqref{eq:hJdef} we can write down the invariants
for the dual map \eqref{eq:Diii} as:
\begin{equation}
    \alpha\Ilow^{\text{\ref{eq:Miii}}}+\beta\Ihigh^{\text{\ref{eq:Miii}}}
    =
    a\Ilow^{\text{\ref{eq:Diii}}} +d \Ihigh^{\text{\ref{eq:Diii}}}.
    \label{eq:derDiii}
\end{equation}
Therefore we obtain the following expressions:
\begin{subequations}
    \begin{align}
        t^{4}\Ilow^{\text{\ref{eq:Diii}}} &
        \begin{aligned}[t]
            &=
            2 \mu (y+z) \beta t^3+\left[2 \alpha y z+\nu (u z+x y-y z) \beta\right] t^2
            \\
            &\quad 
            +2 y z (u z+x y) \beta,
        \end{aligned}
        \label{eq:I1iii}
        \\
        t^{8}\Ihigh^{\text{\ref{eq:Diii}}} &
        \begin{aligned}[t]
            &=
            2 \mu^2 (y z+z^2+y^2) \beta t^6
            \\
            &\quad
            +\left[ y \mu z (y+z) \alpha+\mu \nu (2 u z^2+2 x y^2+y z x+z u y) \beta\right] t^5
            \\
            &
            \quad
            +\frac{1}{2}
            [\nu y z (2 y z+x y+u z-u x) \alpha
            \\
            &\quad
            \phantom{+}+ \nu^2 (4 u x y z+u^2 z^2-3 z^2 y^2+x^2 y^2) \beta] t^4
            \\
            &\quad
            +2 \mu y z (2 u z^2+2 x y^2+y z x+z u y) \beta t^3
            \\
            &
            \quad
            +( z^2 y^2 (u z+x y-u x) \alpha+\nu y z (u z+2 x y) (2 u z+x y) \beta) t^2
            \\
            &\quad 
            +2 y^2 z^2 (u^2 z^2+x^2 y^2+u x y z) \beta.
        \end{aligned}
        \label{eq:I2iii}
    \end{align}
    \label{eq:intDiii}
\end{subequations}
The first invariant \eqref{eq:I1iii} has degree pattern
$\left( 1,2,2,1 \right)$ and the second invariant has degree pattern 
$\left( 2,4,4,2 \right)$. 
However, the degree pattern of the second invariant is not
\emph{minimal}: we can reduce the degree pattern of the second invariant 
to $(1,3,3,1)$ by replacing $\Ihigh$ with  $2\beta \Ihigh-\Ilow^2$.
Moreover, we  can see that the existence of these two invariants is not sufficient
to ensure the low growth of the dual map \eqref{eq:Diii}.
Finally, we note that that the dual map \eqref{eq:Diii} according to the
test in \cite{Gubbiotti_dcov} does not possess a Lagrangian.

\subsection{Maps (iv)}

%The maps ivred\ are the same as the ones of list VI, i.e.
%it is $\dP_\text{I}^{(2)}$.
The main map $\pvec{x}\mapsto\varphi_\text{iv}\left( \pvec{x} \right)=\pvec{x'}$ 
has the following components:
\begin{equation}
    \begin{aligned}
        x'&
        \begin{aligned}[t]
            &=-t^3 a-b t^2 y-d \nu y (x+y+z) t
            \\
            &\quad
            -d y (y^2+2 x y+2 y z+x^2+x z+u z+z^2),
        \end{aligned}
        \\
        y'&=d y x^2, 
        \quad
        z'= d x y^2 ,
        \quad
        u'= d z x y , 
        \quad
        t'= d t x  y.
    \end{aligned}
    \tag{P.iv}
    \label{eq:Miv}
\end{equation}
This map depends on four parameters $a,b, d$ and $\nu$.
We note that the map \eqref{eq:Miv} is the autonomous \ref{dPI2}, 
derived in \cite{CresswellJoshi1999} and whose invariants, duality
and growth properties where studied in \cite{JoshiViallet2017}.
For sake of completeness we repeat these properties here.
The map \eqref{eq:Miv} has the following degrees of iterates:
\begin{equation}
    \begin{aligned}
        \left\{d_n\right\}_{\text{\ref{eq:Miv}}}
        &= 1, 3, 6, 12, 21, 33, 47, 64, 83, 
        \\
        &\phantom{+}104, 128, 154, 183, 214, 248, 284\dots
    \end{aligned}
    \label{eq:degMiv}
\end{equation}
with generating function:
\begin{equation}
    g_{\text{\ref{eq:Miv}}}(s) 
    =-\frac{s^{10}-s^9-s^6+2 s^4+2 s^3+s+1}{(s+1)(s-1)^3}.
    \label{eq:gfMiv}
\end{equation}
This means that the main map is integrable according to the
algebraic entropy test with \emph{quadratic} growth.

The map \eqref{eq:Miv} has the following invariants:
\begin{subequations}
    \begin{align}
        t^{4}\Ilow^{\text{\ref{eq:Miv}}} &
        \begin{aligned}[t]
        &=
        t^3 (y+z) a+z b t^2 y+d \nu y z (y+z) t
        \\
        &\quad
        -d y z (u x-x y-2 y z-u z-y^2-z^2),
        \end{aligned}
        \label{eq:IlowMiv}
        \\
        t^{5}\Ihigh^{\text{\ref{eq:Miv}}} &
        \begin{aligned}[t]
        &=
        -\nu a (y+z) t^4+\left[\left(y^2+z^2+2 y z+u z+x y\right) a-y z b \nu\right] t^3
        \\
        &\quad
        -y z \left[\nu^{2}d (y+z)-b (y+ z+ u+ x)\right] t^2
        +d \nu y z (u y+x z+2 u x) t
        \\
        &\quad 
        +d z y (x+u+z+y) (y^2+z^2+2 y z+u z+x y).
        \end{aligned}
        \label{eq:IhighMiv}
    \end{align}
    \label{eq:intMiv}
\end{subequations}

Using the methods of \cite{Gubbiotti_dcov} we have that the map
\eqref{eq:Miv} is variational.
In affine coordinates $w_{n}$ its Lagrangian is given by:
\begin{equation}
    \begin{aligned}
        L_\text{\ref{eq:Miv}}&=
        w_{{n}}w_{{n+1}}w_{{n+2}}+\frac{w_{{n}}^{3}}{3}
        +w_{{n+1}}w_{{n}}^{2}+w_{{n+1}}^{2}w_{{n}}
        \\
        &\quad
        +\nu \left( \frac{w_{{n}}^{2}}{2}+w_{{n+1}}w_{{n}} \right)
        +\frac {a }{d}\log  \left( w_{{n}} \right) +\frac{b}{d}w_{{n}}.
    \end{aligned}
    \label{eq:LagrPiv}
\end{equation}
Using Corollary~\ref{C:Poisson}, 
we obtain the following non-degenerate Poisson 
bracket\footnote{Asterisked entries are placed to avoid the repetitions of
    entries, since a Poisson-bracket is skew-symmetric $J_{i,j}=-J_{j,i}$.}
\begin{equation}
    J_\text{\ref{eq:Miv}} = 
    \begin {bmatrix} 0 & 0 & {\frac {1}{dw_{{n-1}}}}& 
        -{\frac {\mu+w_{{n-2}}+2 (w_{{n}}+ w_{{n-1}})+w_{{n+1}}}{dw_{{n-1}}w_{{n}}}}
            \\
            0&0&0&{\frac {1}{dw_{{n}}}}
            \\
            -*&0&0&0
            \\ 
            -*&-*
            &0&0
        \end{bmatrix}.
\label{E:Poisson_PIV}
\end{equation}
One can check that the invariants \eqref{eq:intMiv} are in involution
with respect to the Poisson bracket \eqref{E:Poisson_PIV}.
Therefore, the map \eqref{eq:Miv} is Liouville integrable.

%\paragraph{Dual map.}
The dual map $\pvec{x}\mapsto\varphi^{\vee}_\text{iv}\left( \pvec{x} \right)=\pvec{x'}$ 
has the following components:
\begin{equation}
    \begin{aligned}
        x'&
        \begin{aligned}[t]
            &=
            \left[z^2+(y+u-\nu t) z+x (\nu t-x-y)\right] \beta+t \alpha (z-x)
        \end{aligned}
        \\
        y' &=  x^2 \beta, \quad
        z' =  x y \beta, \quad
        u' = x \beta  z, \quad
        t' = t \beta  x.
    \end{aligned}
    \tag{Q.iv}
    \label{eq:Div}
\end{equation}
This map depends on three parameters $\alpha,\beta$, and $\nu$.
The parameter $\nu$ is shared with the main map \eqref{eq:Miv}.
The map \eqref{eq:Div} has the following degrees of iterates:
\begin{equation}
    \begin{aligned}
        \left\{d_n\right\}_{\text{\ref{eq:Div}}}
        &= 1, 2, 4, 7, 11, 17, 24, 32, 41, 52, 64, 77, 91, 107, 124, 
        \\
        &\phantom{+}142,161, 182, 204, 227, 251\dots
    \end{aligned}
    \label{eq:degDiv}
\end{equation}
with generating function:
\begin{equation}
    g_{\text{\ref{eq:Div}}}(s) = -\frac{2 s^5+s^3+s^2+1}{(s+1)(s^2+1)(s-1)^3}.
    \label{eq:gfDiv}
\end{equation}
%The growth is clearly \emph{exponential}.
This means that the dual map is integrable according to the
algebraic entropy test with \emph{quadratic} growth, just like the main map.

Since the main map \eqref{eq:Miv} possesses two invariants and 
depends on $a,b$ and $d$ whereas the dual map \eqref{eq:Div} do not
depend on them according to \eqref{eq:hJdef} we can write down the invariants
for the dual map \eqref{eq:Div} as:
\begin{equation}
    \alpha\Ilow^{\text{\ref{eq:Miv}}}+\beta\Ihigh^{\text{\ref{eq:Miv}}}
    =
    aI_{1}^{\text{\ref{eq:Div}}} +d I_{2}^{\text{\ref{eq:Div}}}
    +b I_{3}^{\text{\ref{eq:Div}}}.
    \label{eq:derDiv}
\end{equation}
Therefore we obtain the following expressions:
\begin{subequations}
    \begin{align}
        t^{2}I_{1}^{\text{\ref{eq:Div}}} &
        \begin{aligned}[t]
            &=
            (y+z) (\alpha-\nu \beta) t+(y^2+z^2+2 y z+u z+x y) \beta,
        \end{aligned}
        \label{eq:I1iv}
        \\
        t^{5}I_{2}^{\text{\ref{eq:Div}}} &
        \begin{aligned}[t]
            &=
            \nu y z (y+z)\left( \alpha-\nu\beta \right) t^2
            \\
            &\quad
            +\left[y z (u y+x z+2 u x) \beta \nu-y z (u x-x y-2 y z-u z-y^2-z^2) \alpha\right] t
            \\
            &\quad
            +y z (x+u+z+y) (y^2+z^2+2 y z+u z+x y) \beta,
        \end{aligned}
        \label{eq:I2iv}
        \\
        t^{3}I_{3}^{\text{\ref{eq:Div}}} &
        \begin{aligned}[t]
            &=
            y z (\alpha-\nu \beta) t+y z (x+u+z+y) \beta.
        \end{aligned}
        \label{eq:I3iv}
    \end{align}
    \label{eq:intDiv}
\end{subequations}
The invariants \eqref{eq:I1iv} and \eqref{eq:I3iv} 
both have degree pattern $\left( 1,2,2,1 \right)$.
However, the second invariant is not minimal and it 
can be replaced with an invariant of degree pattern $(1,3,3,1)$.
Moreover, using the test of \cite{Gubbiotti_dcov},
we obtain that the map \eqref{eq:Div} is not variational.
Therefore, we conclude that the dual map \eqref{eq:Div} is integrable
in the sense of the existence of invariants, i.e. according to definition
\ref{def:invint}.

\subsection{Maps (v)}

%\paragraph{Main map.}
The main map $\pvec{x}\mapsto\varphi_\text{v}\left( \pvec{x} \right)=\pvec{x'}$ 
has the following components:
\begin{equation}
    \begin{aligned}
        x' &=-d \left( x+z \right) ^{2}{y}^{3}
        - \left[ \nu \left( x+z \right) {t}^{2}+u{z}^{2} \right] d{y}^{2}
        - c {t}^{4}y-{t}^{5}a,
        \\
        y' &= d{x}^{3}{y}^{2}, 
        \quad
        z' =d{y}^{3}{x}^{2},
        \quad
        u' = dz{x}^{2}{y}^{2}, 
        \quad
        t'= dt{x}^{2}{y}^{2}.
    \end{aligned}
    \tag{P.v}
    \label{eq:Pv}
\end{equation}
This map depends on the parameters $a,c,d$ and $\nu$.
The map \eqref{eq:Pv} has the following degrees of iterates:
\begin{equation}
    \begin{aligned}
        \left\{d_n\right\}_{\text{\ref{eq:Pv}}}
        &=1,5,15,35,65,103,149,201,261,329,405,489,581,681\dots
    \end{aligned}
    \label{eq:degPv}
\end{equation}
with generating function:
\begin{equation}
    g_{\text{\ref{eq:Pv}}}(s) = -\frac{1+2 s+3 s^2+4 s^3-2 s^5-2 s^7+2 s^8}{(s-1)^3}.
    \label{eq:gfPv}
\end{equation}
This means that the main map is integrable according to the
algebraic entropy test with \emph{quadratic} growth.

The map \eqref{eq:Pv} has the following invariants:
\begin{subequations}
    \begin{align}
        t^{6}\Ilow^{\text{\ref{eq:Pv}}} &
        \begin{aligned}[t]
        &={z}^{2}d \left( x+z \right) {y}^{3}
        +{z}^{2}d \left( \nu{t}^{2}+zu-ux \right) {y}^{2}
        \\
        &\quad
        +{t}^{4} \left(  c z+at \right) y+{t}^{5}az
        \end{aligned}
        \label{eq:IlowPv}
        \\
        t^{8}\Ihigh^{\text{\ref{eq:Pv}}} &
        \begin{aligned}[t]
        &={z}^{2}d \left( x+z \right) ^{2}{y}^{4}
        +2 {z}^{3}ud \left( x+z \right) {y}^{3}
        \\
        &\quad
        + \left\{ {z}^{4}d{u}^{2}
            + \left[  \left(c-{\nu}^{2}d \right) {t}^{2}+2 xu\nu d \right] {t}^{2}{z}^{2}
        \right.
            \\
            &\quad
            \left.\phantom{+}
            + \left( at+x c  \right) {t}^{4}z+{t}^{5}ax \right\} {y}^{2}
            \\
            &\quad
            + \left[  \left( at+u c  \right) {z}^{2}
            -{t}^{2}\nu c z-a\nu{t}^{3} \right] {t}^{4}y
            \\
            &\quad
            +z{t}^{5}a \left( zu -\nu{t}^{2}\right).
        \end{aligned}
        \label{eq:IhighPv}
    \end{align}
    \label{eq:intPv}
\end{subequations}

Using the methods of \cite{Gubbiotti_dcov} we have that the map
\eqref{eq:Pv} is variational.
In affine coordinates $w_{n}$ its Lagrangian is given by:
\begin{equation}
    L_\text{\ref{eq:Pv}}=
    w_{{n}}w_{{n+1}}^{2}w_{{n+2}}
    +\frac{w_{{n+1}}^{2}w_{{n}}^{2}}{2}
    +\nu w_{{n+1}}w_{{n}}
    -{\frac {a}{d}}\frac{1}{w_{{n}}}
    +\frac{c}{d}\log\left( w_{{n}} \right)
    \label{eq:LagrPv}
\end{equation}
Using Corollary \ref{C:Poisson}
we obtain the following non-degenerate Poisson structure
\begin{equation}
    \label{E:Poisson_PV}
    J_\text{\ref{eq:Pv}}=
    \begin{bmatrix} 0&0&
        \frac{1}{w_{{n-1}}^{2}}&
        -{\frac {2 \left(w_{{n}}w_{{n-1}}+w_{{n}}w_{{n+1}}+w_{{n-2}}w_{{n-1}}\right)+\nu}{w_{{n-1}}^{2}w_{{n}}^{2}}}
        \\ 
        0&0&0&\frac{1}{w_{{n}}^{2}}
        \\ 
        -*&0&0&0
        \\
        -*&-*&0&0\end {bmatrix}. 
\end{equation}
One can check that the invariants \eqref{eq:intPv} are in involution
with respect to the Poisson bracket \eqref{E:Poisson_PV}.
Hence,  the map \eqref{eq:Pv} is Liouville integrable.

%\paragraph{Dual map.}
The dual map $\pvec{x}\mapsto\varphi^{\vee}_\text{v}\left( \pvec{x} \right)=\pvec{x'}$ 
has the following components:
\begin{equation}
    \begin{aligned}
        x'&
        \begin{aligned}[t]
        &= 
        \left[ \nu \left( x-z \right) {t}^{2} + \left( u+y \right) {z}^{2}-{x}^{2}y \right] \beta
        -{t}^{2}\alpha  \left( x-z \right),
        \end{aligned}
        \\
        y' &= \beta{x}^{3},
        \quad
        z' = \beta{x}^{2}y, 
        \quad
        u' = \beta{x}^{2}z,
        \quad
        t' = \beta{x}^{2}t.
    \end{aligned}
    \tag{Q.v}
    \label{eq:Qv}
\end{equation}
This map depends on three parameters $\alpha,\beta$ and $\nu$.
The parameter $\nu$ is shared with the main map \eqref{eq:Pv}.
The map given by \eqref{eq:Qv} has the following degrees of iterates:
\begin{equation}
    \begin{aligned}
        \left\{d_n\right\}_{\text{\ref{eq:Qv}}}
        &=1,3,9,19,33,51,73,99,129,163\dots
    \end{aligned}
    \label{eq:degQv}
\end{equation}
with generating function:
\begin{equation}
    g_{\text{\ref{eq:Qv}}}(s) = -\frac{3s^2+1}{(s-1)^3}.
    \label{eq:gfQv}
\end{equation}
This means that the main map is integrable according to the
algebraic entropy test with \emph{quadratic} growth like the main map.

The main map \eqref{eq:Pv} possesses two invariants and 
depends on $a,c$ and $d$ whereas the dual map \eqref{eq:Qv} do not
depend on them. Then according to \eqref{eq:hJdef} we can write down the invariants
for the dual map \eqref{eq:Qv} as:
\begin{equation}
    \alpha\Ilow^{\text{\ref{eq:Pv}}}+\beta\Ihigh^{\text{\ref{eq:Pv}}}
    =
    a I_{1}^{\text{\ref{eq:Qv}}} +c I_{2}^{\text{\ref{eq:Qv}}}
    +d I_{3}^{\text{\ref{eq:Qv}}}.
    \label{eq:derQv}
\end{equation}
Therefore we obtain the following expressions:
\begin{subequations}
    \begin{align}
        t^{4}I_{1}^{\text{\ref{eq:Qv}}} &=
               \beta \left[  \left( x+z \right) {y}^{2}+y{z}^{2}+u{z}^{2} \right] t 
            +(\alpha-\beta\nu)( y+z) {t}^{3},
        \label{eq:I1Qv}
        \\
        t^{4}I_{2}^{\text{\ref{eq:Qv}}} &
        =\left\{  \left[  \left( z+x \right) y+zu-\nu{t}^{2} \right] \beta+{t}^{2}\alpha \right\} zy,
        \label{eq:I2Qv}
        \\
        t^{8}I_{3}^{\text{\ref{eq:Qv}}} &
        \begin{aligned}[t]
            &=
            {y}^{2}{z}^{2} \left\{  
                \left[ \left( x+z \right) ^{2}{y}^{2}+2 uz \left( x+z \right) y
                -{\nu}^{2}{t}^{4}+2 \nu{t}^{2}ux+{u}^{2}{z}^{2} \right]\beta
            \right.
            \\
            &\qquad \quad
            \phantom{+}+ \left[  \left( x+z \right) y+\nu{t}^{2}-ux+zu \right] {t}^{2}\alpha 
            \Bigr\}.
        \end{aligned}
        \label{eq:I3Qv}
    \end{align}
    \label{eq:intQv}
\end{subequations}
We note that the degree pattern of these invariants
is $\left( 1,2,2,1 \right)$, $\left( 1,2,2,1 \right)$ and
$\left( 2,4,4,2 \right)$ respectively.
Finally, using the test of \cite{Gubbiotti_dcov},
we obtain that the map \eqref{eq:Qv} is not variational.
Therefore, we conclude that the dual map \eqref{eq:Qv} is integrable
in the sense of the existence of invariants, i.e. according to definition
\ref{def:invint}.

\stepcounter{eqlistred}
\subsection{Maps (vi)}

%\paragraph{Main map.}

The main map $\pvec{x}\mapsto\varphi_\text{vi}\left( \pvec{x} \right)=\pvec{x'}$ 
has the following components:
\begin{equation}
    \begin{aligned}
        x' &
        \begin{aligned}[t]
            &=-\delta a{t}^{5}
            -\delta  \left[  \left( u-y \right) \mu a\delta+cy+d \left( x+z \right)  \right] {t}^{4}
            \\
            &\quad
            + \left\{ a\mu  \left[ u{y}^{2}+\left( x+z \right) ^{2}y+u{z}^{2} \right] \delta+d \left( x+z
            \right) {y}^{2} \right\} {t}^{2}
             \\
             &\quad
             -\mu  \left[  \left( x+z \right) ^{2}y+u{z}^{2} \right] a{y}^{2}
        \end{aligned}
        \\
        y' &= a\mu x \left( \delta {t}^{2}-{y}^{2} \right)  \left( \delta {t}^{2}-{x}^{2} \right), 
        \quad
        z' = a\mu y \left( \delta {t}^{2}-{y}^{2} \right)  \left( \delta {t}^{2}-{x}^{2} \right),
        \\
        u' &= a\mu z \left( \delta {t}^{2}-{y}^{2} \right)  \left( \delta {t}^{2}-{x}^{2} \right), 
        \quad
        t'= a\mu  t\left( \delta {t}^{2}-{y}^{2} \right)  \left( \delta {t}^{2}-{x}^{2} \right) . 
    \end{aligned}
    \tag{P.vi}
    \label{eq:Pvi}
\end{equation}
This map depends on the five parameters $a,c,d$ and $\mu,\delta$.
The maps \eqref{eq:Pvi} it is a slight generalization of \ref{dPII2} equation
which was discussed in \cite{JoshiViallet2017}. 
Here, we recall its properties and we discuss its duality in the parameter space.
First, the map \eqref{eq:Pvi} has the following degrees of iterates:
\begin{equation}
    \begin{aligned}
        \left\{d_n\right\}_{\text{\ref{eq:Pvi}}}
        &=1,5,15,35,65,103,149,201,261,329,405,489,581,681\dots
    \end{aligned}
    \label{eq:degPvi}
\end{equation}
with generating function:
\begin{equation}
    g_{\text{\ref{eq:Pvi}}}(s) = -\frac{1+2 s+3 s^2+4 s^3-2 s^5-2 s^7+2 s^8}{(s-1)^3}.
    \label{eq:gfPvi}
\end{equation}
This means that the main map is integrable according to the
algebraic entropy test with \emph{quadratic} growth.

The map \eqref{eq:Pvi} has the following invariants:
\begin{subequations}
    \begin{align}
        t^{6}\Ilow^{\text{\ref{eq:Pvi}}} &
        \begin{aligned}[t]
        &=a\delta  \left( y+z \right) {t}^{5}
        - \left[  \left( u \left( x-z\right) -xy \right) \mu a\delta
        -cyz-d\left({y}^{2}+{z}^{2}\right) \right] \delta {t}^{4}
        \\
        &\quad
        - \left\{ d{y}^{2}{z}^{2}
            +a\delta\mu\left( {y}^{2}+{z}^{2} \right) \left[\left( x+z \right)y-\left( x-z \right)u   \right]
        \right\} {t}^{2}
             \\
             &\quad
             +a\mu{y}^{2}{z}^{2} \left[ (x+z)y- (x-z)u\right]
        \end{aligned}
        \label{eq:IlowPvi}
        \\
        t^{8}\Ihigh^{\text{\ref{eq:Pvi}}} &
        \begin{aligned}[t]
        &=
        {\delta}^{2}a\left( u+x+y+z \right) {t}^{7}
        \\
        &\quad
        +{\delta}^{2} \left[ a
        \delta \mu  \left( {u}^{2}-uy+{x}^{2}-xz+{y}^{2}+{z}^{2} \right)\right.
        \\
        &\qquad
        +\left.\left( cu+dx+dz \right) y+ \left( cx+du \right) z+xdu \right] {t}^{6}
        \\
        &\quad
        -\delta a \left[  \left( x+z \right) {y}^{2}+y{z}^{2}+u{z}^{2}\right] {t}^{5}
         \\
         &\quad
         -\delta  \left\{  \left[  \left( {u}^{2}+2 {x}^{2}+xz
         +{z}^{2} \right) {y}^{2}
        +z \left( 2x+z \right) uy
        +{z}^{2} \left( 2{u}^{2}+{x}^{2} \right)  \right] \mu \delta a
        \right.
        \\
        &\qquad
        +d \left( x+z\right) {y}^{3}+ \left( x+z \right)  \left( cz+du \right) {y}^{2}
        \\
        &\qquad\left.
        +{z}^{2} \left( cu+dx+dz \right) y+du{z}^{2} \left( x+z \right)  \right\} {t}^{4}
        \\
        &\quad
        + \biggl\{ 
            2 \mu a\delta \left[ \frac{1}{2}  \left( x+z \right) ^{2}{y}^{4}
                +uz\left( x+z \right) {y}^{3}
                +\frac{ {u}^{2}{z}^{4}}{2}
            \right.
                \\
                &\quad\quad\left.
                +{z}^{2} \left( {u}^{2}+{x}^{2}+xz+\frac{{z}^{2}}{2} \right) {y}^{2}
            +u{z}^{3} \left( x+z \right) y\right]
         \\
         &\qquad
         +d{y}^{2}{z}^{2} \left( x+z \right)  \left( u+y\right) \biggr\} {t}^{2}
         - \left[  \left( x+z \right) y+uz \right] ^{2}\mu a {z}^{2}{y}^{2}        
        \end{aligned}
        \label{eq:IhighPvi}
    \end{align}
    \label{eq:intPvi}
\end{subequations}

Using the methods of \cite{Gubbiotti_dcov} we have that the map
\eqref{eq:Pvi} is variational.
In affine coordinates $w_{n}$ its Lagrangian is given by:
\begin{equation}
    \begin{aligned}
        L_\text{\ref{eq:Pvi}}&=
        \left( w_{{n+1}}^{2}-\delta \right) w_{{n}}w_{{n+2}}
    +\frac{w_{{n+1}}^{2}w_{{n}}^{2}}{2}-{\frac {d}{a\mu}}  w_{{n+1}}w_{{n}}  
    \\
    &\quad -{\frac {1}{2a\mu} \left[ \delta  \left( \delta a\mu-c \right) \log  \left( w_{{n}}^{2}-\delta \right) 
    +2 a\sqrt {\delta}\arctanh \left({\frac {w_{{n}}}{\sqrt {\delta}}}\right) \right] }
    \end{aligned}
    \label{eq:LagrPvi}
\end{equation}
 Using Corollary \ref{C:Poisson}, we obtain the following non-degenerate Poisson structure for the map
\eqref{eq:Pvi}
\begin{equation}
\label{E:Poisson_PVI}
J_\text{\ref{eq:Pvi}} = 
\begin{bmatrix} 
0&0&\frac{1}{ w_{{n-1}}^{2}-\delta} &-{\frac {2 a\mu\ \left(w_{{n}}w_{{n-1}}+w_{{n}}w_{{n+1}}w_{{n-1}}w_{{n-2}}\right)-d}{a\mu  \left(\delta -w_{{n-1}}^{2} \right)  \left( \delta-w_{{n}}^{2} \right) }}
\\
0&0&0&\frac{1}{w_{n}^{2} -\delta}
\\
-*&0&0&0
\\ 
-*&-*&0&0
\end {bmatrix}.
\end{equation} 
One can check that the invariants \eqref{eq:intPvi} are in involution
with respect to the Poisson bracket \eqref{E:Poisson_PVI}.
Therefore, the map \eqref{eq:Pvi} is Liouville integrable.

%\paragraph{Dual map.}
The dual map $\pvec{x}\mapsto\varphi^{\vee}_\text{vi}\left( \pvec{x} \right)=\pvec{x'}$ 
has the following components:
\begin{equation}
    \begin{aligned}
        x'
        &=\left[ \delta\,{t}^{2}u-\left( y+u \right) {z}^{2}+{x}^{2}y\right] \beta
        +\alpha\,{t}^{2} \left( x-z \right)        
        \\
        y' &= \beta x \left( \delta\,{t}^{2}-{x}^{2} \right) ,
        \quad
        z' = \beta y \left( \delta\,{t}^{2}-{x}^{2} \right), 
        \\
        u' &= \beta z \left( \delta\,{t}^{2}-{x}^{2} \right),
        \quad
        t' = \beta t \left( \delta\,{t}^{2}-{x}^{2} \right).
    \end{aligned}
    \tag{Q.vi}
    \label{eq:Qvi}
\end{equation}
This map depends on three parameters $\alpha,\beta$ and $\delta$.
The parameter $\delta$ is shared with the main map \eqref{eq:Pvi}.
The map given by \eqref{eq:Qvi} has the following degrees of iterates:
\begin{equation}
    \begin{aligned}
        \left\{d_n\right\}_{\text{\ref{eq:Qvi}}}
        &=1,3,9,19,33,51,73,99\dots
    \end{aligned}
    \label{eq:degQvi}
\end{equation}
with generating function:
\begin{equation}
    g_{\text{\ref{eq:Qvi}}}(s) = -\frac{3s^2+1}{(s-1)^3}.
    \label{eq:gfQvi}
\end{equation}
This means that the main map is integrable according to the
algebraic entropy test with \emph{quadratic} growth like the main map.

The main map \eqref{eq:Pvi} possesses two invariants and 
depends on $a,c$ and $d$ whereas the dual map \eqref{eq:Qvi} do not
depend on them. Then according to \eqref{eq:hJdef} we can write down the invariants
for the dual map \eqref{eq:Qvi} as:
\begin{equation}
    \alpha\Ilow^{\text{\ref{eq:Pvi}}}+\beta\Ihigh^{\text{\ref{eq:Pvi}}}
    =
    a I_{1}^{\text{\ref{eq:Qvi}}} +c I_{2}^{\text{\ref{eq:Qvi}}}
    +d I_{3}^{\text{\ref{eq:Qvi}}}.
    \label{eq:derQvi}
\end{equation}
Therefore we obtain the following expressions:
\begin{subequations}
    \begin{align}
        t^{8}I_{1}^{\text{\ref{eq:Qvi}}} &
        \begin{aligned}[t]
            &=\delta  \left[ \beta  \left( u+x+y+z \right) \delta-\alpha  \left( y+z \right)  \right] {t}^{7}
            \\
            &\quad
            +{\delta}^{2}\mu \left\{ \beta  \left( {u}^{2}-uy+{x}^{2}-xz+{y}^{2}+{z}^{2} \right) \delta\right.
                \\
                &\quad\quad\left.+ \left[u \left( x-z \right)-xy  \right] \alpha \right\} {t}^{6}
            \\
            &\quad
            -\beta \delta  \left[ \left( x+z \right) {y}^{2}+y{z}^{2}+u{z}^{2} \right] {t}^{5}
            \\
            &\quad
            -\delta \mu  
            \left\{  
                \beta\delta\left[  
            \left( {u}^{2}+2 {x}^{2}+xz+{z}^{2} \right) {y}^{2}\right.\right.
                    \\
                    &\quad\quad\quad\left.+z \left( 2x+z \right) uy
            +{z}^{2} \left( 2 {u}^{2}+{x}^{2}\right)  \right]
                    \\
                    &\quad\quad\left.
                    +\alpha \left( {y}^{2}+{z}^{2} \right) 
                    \left[  u \left( x-z \right)-\left( x+z \right) y  \right] 
                \right\} {t}^{4}
             \\
             &\quad+2 \mu  
             \left\{  \beta\delta
                 \left[ 
                    \frac{{y}^{4}}{2}  \left( x+z \right) ^{2}
                     +uz \left( x+z \right) {y}^{3}
                     +u{z}^{3} \left( x+z \right) y
                 \right.\right.
                     \\
                &\qquad\qquad\left.
                    +{z}^{2} \left( {u}^{2}+{x}^{2}+xz+\frac{{z}^{2}}{2} \right) {y}^{2}
                    +\frac{{u}^{2}{z}^{4}}{2}
                 \right]
                     \\
                     &\qquad\qquad\left.
                         +\frac{\alpha {y}^{2}z^{2}}{2}  
                         \left[  u \left( x-z \right) -\left( x+z \right) y \right]
                 \right\} {t}^{2}
                 \\
                 &\quad
            -\left[  \left( x+z \right) y+uz \right]^{2}\beta {z}^{2}\mu {y}^{2}            
        \end{aligned}
        \label{eq:I1Qvi}
        \\
        t^{4}I_{2}^{\text{\ref{eq:Qvi}}} &=
            \beta\left[ u \left( \delta {t}^{2}-{z}^{2} \right) y+{t}^{2}xz\delta 
            -z \left( x+z \right) {y}^{2}\right]
        -\alpha yz{t}^{2} 
        \label{eq:I2Qvi}
        \\
        t^{6}I_{3}^{\text{\ref{eq:Qvi}}} &
        \begin{aligned}[t]
            &=
            \beta\left( \delta {t}^{2}-{z}^{2} \right)  \left( \delta {t}^{2}-{y}^{2} \right)  
            \left( x+z \right)  \left( u+y \right)
            \\
            &\quad
           -\alpha\left[ \delta  \left( {y}^{2}+{z}^{2} \right) {t}^{2}-{y}^{2}{z}^{2} \right]t^{2}
        \end{aligned}
        \label{eq:I3Qvi}
    \end{align}
    \label{eq:intQvi}
\end{subequations}
We note that the degree pattern of these invariants
is $\left( 2,4,4,2 \right)$, $\left( 1,2,2,1 \right)$ and
$\left( 1,2,2,1 \right)$ respectively.
Finally, using the test of \cite{Gubbiotti_dcov},
we obtain that the map \eqref{eq:Qvi} is not variational.
Therefore, we conclude that the dual map \eqref{eq:Qvi} is integrable
in the sense of the existence of invariants, i.e. according to definition
\ref{def:invint}.

\section{Summary and outlook}
\label{sec:conclusions}

In this paper we presented the (P,Q) class of four-dimensional maps.
These maps were obtained by assuming they possess two invariants
satisfying the conditions A, B and C given in section \ref{sec:theproblem}.
In section \ref{sec:theclass} we discussed the integrability properties
of these maps.

Integrability in the (P,Q) list can arise in different ways depending
weather the map is variational or not.
Variational maps are all Liovuille integrable, as remarked in
section \ref{sec:setting}.
%
%For variational maps using the method in \cite{Gubbiotti_dcov}
%we were able to find a Lagrangian and construct an invariant
%full-rank Poisson structure, see section \ref{sec:setting}.
%This implies that the variational maps are integrable in the
%sense of Liouville.
The only additional structure needed for integrability was then the
Lagrangian, constructed using the method in \cite{Gubbiotti_dcov}.
On the other hand integrability in the non-variational maps can arise
in two different ways.
The pair of maps \eqref{eq:Mi} and \eqref{eq:Di} possessing cubic
growth is \emph{deflatable}.
This means that the two maps arise as non-invertible non-local
transformation from two lower-dimensional maps.
In \cite{GJTV_sanya} we proved that the invariants are preserved
in this process and that the integrability of the three-dimensional
maps can be understood using the definition of Liouville integrability
with a rank two Poisson structure.
All the other maps possess quadratic growth and possess a third invariant
of motion.
In the case of the map \eqref{eq:Miii} the third invariant was found by direct
inspection, while in all the other cases it was produced directly from
the duality approach.
As last remark, we note that the maps with three invariants
admit three different degenerate Poisson structure constructed
using the method of \cite{ByrnesHaggarQuispel1999}, but 
this construction does not yield Liouville integrability.

All the remaining maps have exponential growth and are therefore non
integrable in the sense of the algebraic entropy.
Direct search of invariants for these maps excluded the their
existence up to order 14.
Moreover, using the test of \cite{Gubbiotti_dcov}, we proved that
these exponentially-growing maps are not variational.
Therefore we have a strong evidence of the fact that these maps
do not possess any non-degenerate Poisson structure, and therefore
these cannot be Liouville integrable.
Unfortunately, this result is not enough for a complete proof of the
fact that no non-degenerate Poisson structure exists at all.
This is because, in principle, a fourth-order recurrence relation
can be cast into a system of two second-order recurrence relations
which can be variational.
Therefore, as we did in \cite{GJTV_sanya}, we conjecture that the
maps \eqref{eq:Pii} and \eqref{eq:Diii} either do not admit \emph{any}
full-rank Poisson structure, or for all full-rank Poisson structure they
admit their invariants do not commute.

In table \ref{tab:prop} we give a schematic resume of all the above 
considerations.

\begin{table}[hbt]
    \centering
    \begin{tabular}{cccc}
        \toprule
        Equation & Degree pattern of invariants & Degree of growth & Variational
        \\
        \bottomrule
        (P.i)\textsuperscript{*} & (1,3,3,1), (2,4,4,2) & cubic & no
        \\
        (Q.i)\textsuperscript{*} & (1,2,2,1), (2,4,4,2) & cubic & no
        \\
        \midrule
        (P.ii) & (1,2,2,1), (1,3,3,1) & exponential & no
        \\
        (Q.ii) & (2,4,4,2) & exponential & no
        \\
        \midrule
        (P.iii) & (1,3,3,1), (2,4,4,2), (2,5,5,2) & quadratic & no
        \\
        (Q.iii) & (1,2,2,1), (2,4,4,2) & exponential & no
        \\
        \midrule
        (P.iv) ($\dP_\text{I}^{(2)}$) & (1,3,3,1), (2,4,4,2) & quadratic & yes
        \\
        (Q.iv) & (1,2,2,1), (1,2,2,1), (2,4,4,2)  & quadratic & no
        \\
        \midrule
        (P.v)  & (1,3,3,1), (2,4,4,2) & quadratic & yes
        \\
        (Q.v) & (1,2,2,1), (1,2,2,1), (2,4,4,2)  & quadratic & no
        \\
        \midrule
        (P.vi) ($\dP_\text{II}^{(2)}$) & (1,3,3,1), (2,4,4,2) & quadratic & yes
        \\
        (Q.vi) & (1,2,2,1), (1,2,2,1), (2,4,4,2) & quadratic & no
        \\
        \bottomrule
        \multicolumn{4}{l}{{\footnotesize\textsuperscript{*} Deflatable to a three-dimensional
        Liouville integrable map \cite{GJTV_sanya}.}}
    \end{tabular}
    \caption{Integrability properties of the (P,Q) maps.}
    \label{tab:prop}
\end{table}

%Based on the above considerations we underline that 
The search procedure carried out in this paper has been very fruitful giving
some interesting and non-trivial examples of four-dimensional maps.
Indeed, all the maps, but four are new.
%we found were known in the 
%literature, making this kind of approach very fruitful.
Particularly interesting is the variety of behaviours we encountered
in the maps of the class (P,Q).
Work is in progress to characterize the surfaces generated by the 
invariants in both integrable and non-integrable cases. 
We expect this to give some hints on how the integrability arises
from purely geometrical considerations.
This is well known for maps in two dimension with the theory of
elliptic fibrations applied to the QRT mapping \cite{Duistermaat2011book}.
However, it was discussed in \cite{JoshiViallet2017,GJTV_sanya}
how examples with cubic growth can go beyond the existence of
elliptic fibrations making the underlying geometrical structure 
more complex and richier.

Finally, we believe that the direct search of maps with
invariants alongside with the algorithmic tests available in the
discrete setting may produce many new results and integrable
maps in the next years.
Analogous procedure in the continuous case
still yield many new result after more that fifty years of
their introduction \cite{Fris1965,PostWinternitz2011,EscobarWinternitzYurdusen2018}.
Work is in progress to extend the present class by considering
invariants of more general form.

\section{Acknowledgments}

This research was supported by an Australian Laureate Fellowship \#FL120100094 
and grant \#DP160101728 from the Australian Research Council.

\bibliographystyle{plain}
\bibliography{bibliography}

\end{document}